\documentclass[UKenglish,runningheads,envcountsame]{llncs}

\usepackage{multirow}
\usepackage{hyperref}
\usepackage{xspace}
\usepackage{amsmath,amssymb}
\usepackage{stmaryrd} 

\usepackage{apxproof} 
\usepackage{algorithm}
\usepackage{algpseudocode}
\usepackage{paralist}

\usepackage{todonotes} 

\usepackage{tikz-cd}
\usepackage{wrapfig}
\usepackage{amsmath}    
\usepackage{lipsum}     

\usepackage{tikz}
 \usetikzlibrary{trees}
 \usetikzlibrary{shapes}
 \usetikzlibrary{fit}
 \usetikzlibrary{shadows}
 \usetikzlibrary{backgrounds}
 \usetikzlibrary{arrows,automata}

\tikzset{
   n/.style= {circle,fill,inner sep=1.5pt,node distance=2cm}
  ,acc/.style={circle,draw,inner sep=3pt,node distance=2cm}
  ,phantom/.style={circle},
  ,arr/.style={->, >=stealth, semithick, shorten <= 3pt, shorten >= 3pt}
}

\makeatletter
\newcommand\mysubsec{\@startsection{paragraph}{4}{\z@}%
  {-6\p@ \@plus -4\p@ \@minus -4\p@}%
  {-0.5em \@plus -0.22em \@minus -0.1em}%
  {\normalfont\normalsize\bfseries}}
\makeatother

\spnewtheorem{assumptions}[theorem]{Assumptions}{\bfseries}{\rmfamily}
\spnewtheorem{notation}[theorem]{Notation}{\bfseries}{\rmfamily}
\spnewtheorem{observation}[theorem]{Observation}{\bfseries}{\rmfamily}
\spnewtheorem{defn}[theorem]{Definition}{\bfseries}{\rmfamily}
\spnewtheorem{alg}[theorem]{Algorithm}{\bfseries}{\rmfamily}
\spnewtheorem{expl}[theorem]{Example}{\bfseries}{\rmfamily}
\spnewtheorem{rem}[theorem]{Remark}{\bfseries}{\rmfamily}
\spnewtheorem{fact}[theorem]{Fact}{\bfseries}{\rmfamily}
\spnewtheorem{construction}[theorem]{Construction}{\bfseries}{\rmfamily}
\spnewtheorem{examples}[theorem]{Examples}{\bfseries}{\rmfamily}

\renewcommand{\Box}{\square}

\newcommand{\sem}[1]{\llbracket #1 \rrbracket}

\newcommand{\takeout}[1]{\empty}

\newcommand{\A}{\mathcal{A}}
\newcommand{\AP}{\mathsf{AP}}
\newcommand{\Inf}{\mathsf{Inf}}
\newcommand{\Fin}{\mathsf{Fin}}
\newcommand{\strat}{\sigma}
\newcommand{\NSA}{\mathcal{N}}
\newcommand{\DSA}{\mathcal{D}}
\newcommand{\Lang}{\mathcal{L}}
\newcommand{\symb}{\mathsf{symb}}

\newcommand{\sys}{\exists}

\begin{document}

\title{Symbolic Solution of Emerson-Lei Games for Reactive Synthesis}
\titlerunning{El-Synthesis} 
\author{%
  Daniel Hausmann, Mathieu Lehaut and Nir Piterman
}
\authorrunning{D.~Hausmann and M.~Lehaut and N.~Piterman}

\institute{University of Gothenburg, Gothenburg, Sweden
}

\maketitle
\begin{abstract}

Emerson-Lei conditions have recently attracted attention due to their succinctness and compositionality properties. In the current work, we show how infinite-duration games with Emerson-Lei objectives can be analyzed in two different ways. First, we show that the Zielonka tree of the Emerson-Lei condition gives rise naturally to a new reduction to parity games. This reduction, however, does not result in optimal analysis. Second, we show based on the first reduction (and the Zielonka tree) how to provide a direct fixpoint-based characterization of the winning region. The fixpoint-based characterization allows for symbolic analysis. It generalizes the solutions of games with known winning conditions such as B\"uchi, GR[1], parity, Streett, Rabin and Muller objectives, and in the case of these conditions reproduces previously known symbolic algorithms and complexity results.

We also show how the capabilities of the proposed algorithm can be exploited in reactive synthesis, suggesting a new expressive fragment of LTL that can be handled symbolically. Our fragment combines a safety specification and a liveness part. The safety part is unrestricted and the liveness part allows to define Emerson-Lei conditions on occurrences of letters. The symbolic treatment is enabled due to the simplicity of determinization in the case of safety languages and by using our new algorithm for game solving. This approach maximizes the number of steps solved symbolically in order to maximize the potential for efficient symbolic implementations.

\end{abstract}

\section{Introduction}
\label{sec:intro}

Infinite-duration two-player games are a strong tool that has been
used, notably, for reactive synthesis from temporal specifications
\cite{PnueliR89}.
Many different winning conditions are considered in the literature.

Emerson-Lei (EL) conditions were initially suggested in the context of
automata but are among the most general (regular) winning conditions
considered for such games.
They succinctly express general liveness properties
by encoding Boolean combinations of events that should occur infinitely or
finitely often.
Automata and games in which acceptance or winning is defined by
Emerson-Lei conditions have garnered attention in recent years~\cite{MuellerSickert17,RenkinDP20,JohnJBK22,HunterD05}, in particular because of their succinctness and good compositionality properties (Emerson-Lei objectives are closed under conjunction, disjunction, and negation).
In this work, we show how infinite-duration two-player games with
Emerson-Lei winning conditions can be solved symbolically.

It has been established that
solving Emerson-Lei games is \textsc{PSpace}-complete and that an exponential amount of memory
may be required by winning strategies~\cite{HunterD05}.
Zielonka trees are succinct tree-representations of $\omega$-regular
winning objectives~\cite{Zielonka98}.
They have been used to obtain tight bounds on the
amount of memory needed for winning in Muller
games~\cite{DziembowskiJW97}, and can also be applied to analyze
Emerson-Lei objectives and games.
One indirect way to solve Emerson-Lei games is by transformation to
equivalent parity games using later-appearance-records~\cite{HunterD05}, and then
solving the resulting parity games.
Another, more recent, indirect approach goes through Rabin games
by first extracting history-deterministic Rabin automata from Zielonka
trees and then solving the resulting Rabin games~\cite{CasaresCL22}.
Here, we give a direct symbolic algorithmic solution for Emerson-Lei
games.
We show how the Zielonka tree allows to directly encode the game as a
parity game.
Furthermore, building on this reduction, we show how to construct a
fixpoint equation system
that captures winning in the game.
As usual, fixpoint equation systems are recipes for game solving
algorithms that manipulate sets of states symbolically.

The algorithm that we obtain in this way is adaptive in the sense that
the nesting structure of
recursive calls is obtained directly from
the Zielonka tree of the given winning objective. As the Zielonka tree
is specific to the objective, this means that the algorithm performs
just the fixpoint computations that are required for that specific objective.
In particular, our algorithm instantiates to previously known
fixpoint iteration algorithms in the case that
the objective is a (generalized) B\"uchi, GR[1], parity, Streett, Rabin or Muller
condition, reproducing
previously known algorithms and complexity results.
As we use fixpoint iteration, the instantiation of our algorithm
to parity game solving is not directly a quasipolynomial algorithm.
In the general setting,
the algorithm solves unrestricted Emerson-Lei games with $k$ colors,
$m$ edges and $n$ nodes in time $\mathcal{O}(k!\cdot m\cdot n^k)$ and
yields winning strategies with memory $\mathcal{O}(k!)$.

We apply our symbolic solution of Emerson-Lei games
to the automated construction of safe systems.
The ideas of synthesis of reactive systems from temporal specifications
go back to the early days of computer science \cite{Church63}.
These concepts were modernized and connected to linear temporal logic
(LTL) and finite-state automata by Pnueli and Rosner \cite{PnueliR89}.
In recent years, practical applications in robotics are using this
form of synthesis as part of a framework producing correct-by-design
controllers
\cite{Kress-GazitFP09,BhatiaMKV11,WongpiromsarnTM12,LiuOTM13,MoarrefK20}.



A prominent way to extend the capacity of reasoning about state spaces
is by reasoning \emph{symbolically} about sets of states/paths.
%
In order to apply this approach to reactive synthesis, different
fragments of LTL that allow symbolic game analysis have been
considered.
%
Notably, the GR[1] fragment has been widely used for the
applications in robotics mentioned above
\cite{PitermanPS06,BloemJPPS12}.
But also larger fragments are being considered and experimented with
\cite{Ehlers11,Ehlers11a,SohailS13}.
Recently,
De Giacomo and Vardi suggested that similar advantages can be had by in
changing the usual semantics of
LTL from considering infinite models to finite models (LTL$_f$)
\cite{GiacomoV15}.
The complexity of the problem remains doubly-exponential, however,
symbolic techniques can be applied.
As models are finite, it is possible to use the
classical subset construction (in contrast to B\"uchi determinization),
which can be reasoned about symbolically.
Furthermore, the resulting games have simple reachability objectives.
This approach with finite models is used for applications in
planning \cite{CamachoTMBM17,CamachoM19} and robotics
\cite{BhatiaMKV11}.

Here, we harness our symbolic solution to Emerson-Lei games to
suggest a large fragment of LTL that can be reasoned about
symbolically.
We introduce the \emph{Safety and Emerson-Lei} fragment
whose formulas are conjunctions
$\varphi_{\mathrm{safety}}\land\varphi_{\mathrm{EL}}$
between an (unrestricted) safety condition and an
(unrestricted) Emerson-Lei condition defined in terms of game states.
This fragment generalizes GR[1] and the previously mentioned works in \cite{Ehlers11,Ehlers11a,SohailS13}.
We approach safety and Emerson-Lei LTL synthesis in two steps: first,
consider only the safety part and convert it to a symbolic safety
automaton;
second, reason symbolically on this automaton
by solving Emerson-Lei games using our novel symbolic algorithm.

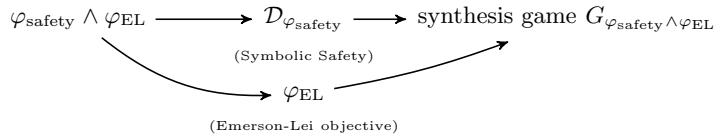
\begin{center}

\begin{footnotesize}
\begin{tikzpicture}[node distance=1.0cm,>=stealth',auto,semithick,        shorten > = 1pt]

  \tikzstyle{place}=[circle,thick,draw=blue!75,fill=blue!20,minimum size=6mm]
  \tikzstyle{red place}=[place,draw=red!75,fill=red!20]
  \tikzstyle{transition}=[rectangle,thick,draw=black!75,
  			  fill=black!20,minimum size=4mm]

    \tikzstyle{every state}=[
        draw = black,
        thick,
        fill = white,
        minimum size = 4mm
    ]

    \node[label={below: \footnotesize{}}] (1) at (0,1) {$\varphi_{\mathrm{safety}}\land \varphi_{\mathrm{EL}}$};
    \node[label={below: \tiny{(Symbolic Safety)}}] (3) at (3,1)
    {$\mathcal{D}_{\varphi_{\mathrm{safety}}}$};
    \node (4) at (6.5,1) {synthesis game
    $G_{\varphi_{\mathrm{safety}}\land \varphi_{\mathrm{EL}}}$};
    \node[label={below: \tiny{({Emerson-Lei} objective)}}] (5) at (3,0)
    {$\varphi_{\mathrm{EL}}$};

    \path[->] (1) edge  node  {} (3);
    \path[->] (1) edge  [bend right=20] node   {} (5);
    \path[->] (3) edge  node  {} (4);
    \path[->] (5) edge  [bend right=5] node  {} (4);

\end{tikzpicture}
\end{footnotesize}
\end{center}
\vspace{-5pt}
We show that realizability of a safety and
Emerson-Lei formula $\varphi_{\mathrm{safety}}\land
\varphi_{\mathrm{EL}}$ can be checked in time $2^{\mathcal{O}(m\cdot\log m\cdot 2^n)}$,
where $n=|\varphi_{\mathrm{safety}}|$ and $m=|\varphi_{\mathrm{EL}}|$.
The overall procedure therefore is doubly-exponential in the size of the safety part but only single-exponential in the size of the liveness part;
notably, 
both the automaton determinization and game solving parts can be
implemented symbolically.

We begin by recalling Emerson-Lei games and Zielonka trees
in Section~\ref{sec:el-ziel}, and also prove an upper bound on the size of Zielonka trees. Next we show how to solve Emerson-Lei games by fixpoint computation in Section~\ref{sec:solvingELgames}.
In Section~\ref{sec:synt} we formally introduce the safety and Emerson-Lei fragment of LTL and show how to construct symbolic games with
Emerson-Lei objectives that characterize realizability and that can be solved
using the algorithm proposed in Section~\ref{sec:solvingELgames}.
Omitted proofs and further details can be found in the appendix.
\vspace{-10pt}
\section{Emerson-Lei Games and Zielonka Trees}
\label{sec:el-ziel}
We recall the basics of
Emerson-Lei games \cite{HunterD05} and Zielonka trees~\cite{Zielonka98},
and also show an apparently novel bound on the size of Zielonka trees;
previously, the main interest was on the size of winning strategies
induced by Zielonka trees, which is smaller~\cite{DziembowskiJW97}.

\paragraph{Emerson-Lei games.}
We consider two-player games played between the \emph{existential
player} $\exists$ and its opponent, the \emph{universal player}
$\forall$.
A \emph{game arena} $A=(V, V_\exists,V_\forall,E)$ consists
of a set $V=V_\exists\uplus V_\forall$ of nodes, partitioned
into sets of \emph{existential nodes} $V_\exists$ and \emph{universal nodes}
$V_\forall$, and a set $E\subseteq V\times V$ of \emph{moves}; we put $E(v)=\{v'\in V\mid
(v,v')\in E\}$ for $v\in V$.
A \emph{play} $\pi=v_0 v_1\ldots$ then is a sequence of nodes such that
for all $i\geq 0$, $(v_i,v_{i+1})\in E$; we denote the set of plays
in $A$ by $\mathsf{plays}(A)$.
A \emph{game} $G=(A,\alpha)$ consists of
a game arena together with an objective $\alpha\subseteq
\mathsf{plays}(A)$.

%
A \emph{strategy} for the existential player is a function $\sigma:\, V^*\cdot
V_\exists \to V$ such that for all
$\pi\in V^*$ and $v\in V_\exists$ we have $(v,\sigma(\pi v))\in E$.
A play $v_0 v_1\ldots$ is said to be \emph{compliant} with strategy
$f$ if for all $i\geq 0$ such that $v_i\in V_\exists$ we have
$v_{i+1}=\sigma(v_0\ldots v_i)$.
Strategy $\sigma$ is \emph{winning} for the existential player from node $v\in V$
if all plays starting in $v$ that are compliant
with $\sigma$ are contained in $\alpha$.
We denote by $W_\exists$ the set of nodes winning for the existential player.

In \emph{Emerson-Lei games}, each node is colored by a set
of colors, and the objective $\alpha$ is induced by a formula
that specifies combinations of colors that have to be visited infinitely
often, or are allowed to be visited only finitely often. Formally,
we fix a set $C$ of colors and use \emph{Emerson-Lei formulas}, that is,
finite Boolean formulas
$\varphi\in\mathbb{B}(\{\mathsf{Inf}\,c\}_{c\in C})$ over atoms
of the shape $\mathsf{Inf}\,c$, to define sets of plays.
The satisfaction relation
$\models$ for a set $D\subseteq C$ of colors and an Emerson-Lei formula
$\varphi$
is defined in the usual inductive way, and
we abbreviate $\neg\mathsf{Inf}\,c$ by $\mathsf{Fin}\,c$. E.g. the clauses for atoms
$\mathsf{Inf}\,c$ and negated formulas $\neg\varphi$ are
\vspace{-5pt}
\begin{align*}
	D\models \mathsf{Inf}\, c & \Leftrightarrow c\in D &
	D\models \neg \varphi &\Leftrightarrow D\not\models \varphi
\end{align*}

Consider a game arena $A=(V,V_\exists, V_\forall,E)$.
An \emph{Emerson-Lei condition} is given by an Emerson-Lei formula
$\varphi$
together with a coloring function $\gamma: V\to 2^C$ that
assigns a (possibly empty) set $\gamma(v)$ of colors to each node
$v\in V$.
The formula $\varphi$ and the coloring function $\gamma$ together
specify the objective
\vspace{-5pt}
\begin{align*}
	\alpha_{\gamma,\varphi}=\Big \{v_0v_1\ldots\in\mathsf{plays}(A)
	\Big |
	\{c\in C\mid \forall i.~\exists j \geq i.~c\in\gamma(v_j)\}\models
	\varphi \Big \}
\end{align*}
Thus a play $\pi=v_0v_1\ldots$ is winning for the
existential
player (formally: $\pi\in \alpha_\varphi$)
if and only if the set of colors that are visited infinitely often
by $\pi$ satisfies $\varphi$.
Below, we will also make use of \emph{parity games}, denoted by $(V,V_\exists, V_\forall,E,\Omega)$ where $\Omega:V\to\{1,\ldots, {2k}\}$ is a priority function, assigning priorities to game nodes. The objective of the existential
player then is that the maximal priority that is visited infinitely often is
an even number. Parity games are an instance of Emerson-Lei
games, obtained with set $C=\{p_1,\ldots, p_{2k}\}$ of colors,
a coloring function that assigns exactly one color to each node and
with objective
\begin{equation*}
\mathsf{Parity}(p_1,\ldots, p_{2k})=\textstyle\bigvee_{i\text{ even}}
\left (\mathsf{Inf}\, p_{i} \land \textstyle\bigwedge_{i<j\leq 2k}
\mathsf{Fin}\, p_{j} \right ).
\end{equation*}
Similarly, Emerson-Lei objectives directly encode
(combinations of) other standard objectives,
such as B\"uchi, Rabin, Streett or Muller conditions:
\begin{itemize}
\item[---] $\mathsf{Inf}\, f$ \hfill \textsf{B\"uchi}$(f)$
\item[---] $\textstyle\bigvee_{1\leq i\leq k} \mathsf{Inf}\,e_i\wedge
			\mathsf{Fin}\, f_i$ \hfill \textsf{Rabin}$(e_1,f_1,\ldots,e_k,f_k)$
\item[---] $\textstyle\bigwedge_{1\leq i\leq k} (\mathsf{Inf}\,r_i\to
			\mathsf{Inf}\, g_i)$ \hfill \textsf{Streett}$(r_1,g_1,\ldots,r_k,g_k)$
\item[---] $\textstyle\bigvee_{D\in\mathcal{U}} (\bigwedge_{c\in D}\mathsf{Inf} \,c\wedge\bigwedge_{d\in C\setminus D}\mathsf{Fin}\,d)$
\hfill \textsf{Muller}$(\mathcal{U}\subseteq 2^{C})$
\end{itemize}

\begin{toappendix}
\begin{example}\label{example:elGames} The framework of Emerson-Lei games subsumes many standard types of
	$\omega$-regular games. We show how this works in detail for games with the following objectives.
	\begin{enumerate}
			\item \emph{B\"uchi objective:} Put $C=\{f\}$ and
			\begin{align*}
					\varphi=\mathsf{Inf}\,f.
				\end{align*}
			Then $\pi\in\alpha_\varphi$
			if and only if $\pi$ visits some accepting node $v$ with
			$\gamma(v)=\{f\}$ infinitely often.
			\item \emph{Parity objective:} Put $C=\{p_1,\ldots, p_k\}$ and
			\begin{align*}
					\varphi=\textstyle\bigvee_{i\text{ even}}
		(\mathsf{Inf}\,p_i\wedge
					\bigwedge_{j>i} \mathsf{Fin}\,p_j),
				\end{align*}
			where $i$ and $j$ range over $\{1,\ldots,k\}$.
			Then $\pi\in\alpha_\varphi$
			if and only if there is some even $i$ such that $\pi$ visits $p_i$
			infinitely often,
			and no color $p_j$ with $j>i$ is visited infinitely often by
	$\pi$.
			\item \emph{Rabin objective:} Put $C=\{e_1,f_1,\ldots, e_k,f_k\}$ and
			\begin{align*}
					\varphi=\textstyle\bigvee_{i} (\mathsf{Inf}\,e_i\wedge
					\mathsf{Fin}\, f_i),
				\end{align*}
			where $i$ ranges over $\{1,\ldots, k\}$.
			Then $\pi\in\alpha_\varphi$
			if and only if there is some $i$ such that $\pi$ visits color
			$e_i$ infinitely often
			and does not visit color $f_i$ infinitely often.
			\item \emph{Streett objective:} Put $C=\{r_1,g_1,\ldots, r_k,g_k\}$
	and
			\begin{align*}
					\varphi=\textstyle\bigwedge_{i} (\mathsf{Inf}\, r_i\to
					\mathsf{Inf}\,g_i),
				\end{align*}
			where $i$ ranges over $\{1,\ldots, k\}$.
			Then $\pi\in\alpha_\varphi$
			if and only if for all $i$, if $\pi$ visits color $r_i$
	infinitely
			often,
			then $\pi$ visits color $g_i$ infinitely often.

			\item \emph{Muller objective:} For a Muller condition given as a
			set $W\subseteq\mathcal{P}(C)$ of sets of colors, put
			\begin{align*}
					\varphi=\textstyle\bigvee_{F\in W} (\bigwedge_{c\in F}
					\mathsf{Inf}\,c\wedge
					\bigwedge_{d\notin F} \mathsf{Fin}\,d)
				\end{align*}
			so that $\pi\in\alpha_\varphi$
			if and only $W$ contains some set $F$ of colors such that $\pi$
			visits exactly the
			colors from $F$ infinitely often.
		\end{enumerate}
\end{example}
\end{toappendix}


\paragraph{Zielonka Trees.}

We introduce a succinct encoding of the algorithmic essence of Emerson-Lei objectives
in the form of so-called Zielonka trees~\cite{Zielonka98,DziembowskiJW97}.

\begin{definition}
	The \emph{Zielonka tree} for an Emerson-Lei formula $\varphi$ over
	set $C$ of colors is a tuple
	$\mathcal{Z}_\varphi=(T,R,l)$ where $(T,R\subseteq T\times T)$ is a tree and
	$l:T\to 2^C$ is a labeling function that
	assigns sets of colors $l(t)$ to vertices $t\in T$.
	We denote the root of $(T,R)$ by $r$.
    Then $\mathcal{Z}_\varphi$ is defined to be
	the unique tree (up to reordering of child vertices) that satisfies
	the following constraints.
	\begin{itemize}
		\item The root vertex is labeled with $C$, that is, $l(r)=C$.
		\item Each vertex $t$ has exactly one child vertex $t_D$
		(labeled with $l(t_D)=D$) for
		each set $D$ of colors that is maximal in
$\{D'\subsetneq l(t)\mid D'\models\varphi \Leftrightarrow l(t)\not\models \varphi\}$.
	\end{itemize}
	For $s,t\in T$ such that $s$ is an ancestor of $t$, we write $s\leq
	t$.
	Given a vertex $s\in T$, we denote its set of direct successors by
	$R(s)=\{t\in T\mid (s,t)\in R\}$ and the set of leafs below it by
	$L(s) = \{t\in T \mid s\leq t \text{ and } R(t)=\emptyset\}$; we write
	$L$ for the set of all leafs.
	We assume some fixed total order $\preceq$ on $T$
	that respects $\leq$; this order induces a numbering of $T$.
		A vertex $t$ in the Zielonka tree is said to be \emph{winning} if
	$l(t)\models \varphi$,
	and \emph{losing} otherwise.
	We let $T_\Box$ and $T_\bigcirc$ denote the sets of
	winning and losing vertices in $\mathcal{Z}_\varphi$,
	respectively. Finally, we assign a \emph{level} $\mathsf{lev}(t)$
	to each vertex $t\in T$ so that  $\mathsf{lev}(r)=|C|$,
	and $\mathsf{lev}(s')=\mathsf{lev}(s)-1$ for all $(s,s')\in R$.

\end{definition}

\begin{example}\label{example:zielonkaTrees}
As mentioned above, Emerson-Lei games and Zielonka trees
instantiate naturally to games with, e.g., B\"uchi,
generalized B\"uchi, GR[1], parity, Rabin, Streett and Muller objectives;
for brevity, we illustrate this for selected examples (for more instances, see Example~\ref{example:zielonkaTreesCont} in the appendix).
	\begin{enumerate}
		\item \emph{Generalized B\"uchi condition}:
        Given $k$ colors $f_1,\ldots,f_k$,
		the winning objective $\varphi=\bigwedge_{1\leq i\leq k}\mathsf{Inf}~f_i$
		expresses that all colors are visited infinitely often (not
		necessarily simultaneously); the induced
		Zielonka tree is depicted below with boxes and circles representing winning
		and losing vertices, respectively.\\
\vspace{-10pt}
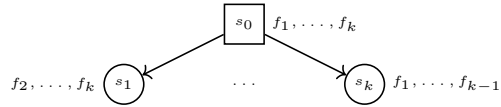
\begin{center}
\tikzset{every state/.style={minimum size=15pt}}
\begin{tiny}
  \begin{tikzpicture}[
		auto,
    node distance=0.8cm,
    semithick
    ]
     \node[state, rectangle, label={right: $f_1,\ldots,f_k$}] (0) {$s_0$};
     \node (yo) [below of=0] {$\ldots$};
     \node (yo1) [left of=yo] {};
     \node (yo2) [right of=yo] {};
     \node[state, label={left: $f_2,\ldots,f_k$}] (1) [left of=yo1] {$s_1$};
     \node[state, label={right: $f_1,\ldots,f_{k-1}$}] (2) [right of=yo2] {$s_k$};
     \path[->] (0) edge node [pos=0.3,left] {} (1);
     \path[->] (0) edge node [pos=0.3,right] {} (2);

  \end{tikzpicture}
\end{tiny}
\end{center}
\vspace*{-5pt}
%

		\item \emph{Streett condition}: The vertices in the Zielonka tree for
		Streett condition given by
		$\varphi=\bigwedge_{1\leq i\leq k}\left (\mathsf{Inf}~r_i
		\to\mathsf{Inf}~g_i \right )$
		are identified by duplicate-free lists $\mathsf{L}$ of colors (each entry being
		$r_i$ or $g_i$ for some $1\leq i\leq k$)
		that encode	the vertex position in
		the tree. Vertex $\mathsf{L}$ has label $l(\mathsf{L})=C\setminus \mathsf{L}$ 
		and is winning if and only if $|\mathsf{L}|$ is even.
		Winning vertices $\mathsf{L}$ have one child vertex
		${\mathsf{L}:g_j}$ for each $g_j\in C\setminus \mathsf{L}$ resulting in
		$|C\setminus \mathsf{L}|/2$ many child nodes. Losing vertices $\mathsf{L}$ 
		have the single
		child vertex $\mathsf{L}:r_j$ where the last entry $\mathsf{last}(\mathsf{L})$ 
		in $\mathsf{L}$	is $g_j$. All leafs are winning and are labeled with
		$\emptyset$. The tree has
		height $2k$ and $2(k!)$ vertices. 

\item
To obtain a Zielonka tree  that has branching at both winning and losing nodes,
we consider the objective
$\varphi_{EL}=(\mathsf{Inf}~a\to
\mathsf{Inf}~b)\wedge((\mathsf{Fin}~a\vee
\mathsf{Fin}~d)\land
\mathsf{Inf}~c)$.
This property can be seen as the conjunction of a
Streett pair $(a,b)$ with two disjunctive Rabin pairs $(c,a)$ and
$(c,d)$, stating that $c$ occurs infinitely often and $a$ occurs
finitely often or $c$ occurs infinitely often and $d$ occurs finitely
often.
In the next page we depict the induced Zielonka tree.

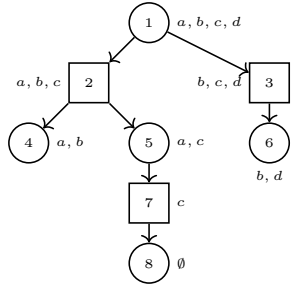
\begin{center}
\tikzset{every state/.style={minimum size=15pt}}
\begin{tiny}
  \begin{tikzpicture}[
		auto,
    node distance=0.8cm,
    semithick
    ]
     \node[state, label={right: $a,b,c,d$}] (0) {$1$};
     \node (yo) [below of=0] {};
     \node (yo2) [right of=yo] {};
     \node[state, rectangle, label={left: $a,b,c$}] (1) [left of=yo] {$2$};
     \node[state, rectangle, label={left: $b,c,d$}] (2) [right of=yo2]
     {$3$};
     \node (yo1) [below of=1] {};
     \node[state, label={right: $a,c$}] (3) [right of=yo1] {$5$};
     \node[state, label={right: $a,b$}] (4) [left of=yo1] {$4$};
     \node[state, label={below: $b,d$}] (5) [below of=2] {$6$};
     \node[state, rectangle, label={right: $c$}] (6) [below of=3] {$7$};
     \node[state, label={right: $\emptyset$}] (7) [below of=6] {$8$};
     \path[->] (0) edge node [pos=0.3,left] {} (1);
     \path[->] (0) edge node [pos=0.3,right] {} (2);
     \path[->] (1) edge node [pos=0.3,left] {} (3);
     \path[->] (1) edge node [pos=0.3,left] {} (4);
     \path[->] (2) edge node [pos=0.3,left] {} (5);
     \path[->] (3) edge node [pos=0.3,left] {} (6);
     \path[->] (6) edge node [pos=0.3,left] {} (7);

  \end{tikzpicture}
\end{tiny}
\end{center}

	\end{enumerate}

\end{example}

\begin{toappendix}
\begin{example}\label{example:zielonkaTreesCont}
We expand on Example~\ref{example:zielonkaTrees} and
provide instances of Zielonka trees
to the objectives from Example~\ref{example:elGames}.
\begin{enumerate}

		\item \emph{B\"uchi condition}: The Zielonka tree for the B\"uchi condition
		$\varphi=\mathsf{Inf}~f$
		consists of just two vertices $r=s_0$ and $s_1$
		labeled with $\{f\}$ and $\emptyset$, respectively. The vertex
		$s_0$ is winning
		and the vertex $s_1$ is losing.

		\item \emph{Parity condition}: Given $k$ priorities
		$\{p_1,\ldots,p_k\}$, the induced
		Zielonka tree of the parity condition $\varphi=\bigvee_{i \text{ even}}
		\mathsf{Inf}~p_i\wedge \bigwedge_{j>i}\mathsf{Fin}~p_j$
		is just a single path of length $k$ in which the
		$i$-th vertex $s_i$ is labeled
		with the set $\{p_j\mid j\leq (k+1)-i\}$; the vertex $s_i$ is
		winning if and
		only if $i$ is even. Assuming that every node has at least one
		priority,
		the leaf is losing and labeled with $\{p_1\}$.

		\item \emph{Rabin condition}: The vertices in the Zielonka tree for Rabin
		conditions
		$\langle e_i,f_i\rangle_{1\leq i\leq k}$
		are identified by duplicate-free lists $\mathsf{L}$ of colors that encode
		the vertex position in
		the tree. Vertices $\mathsf{L}$ have labels
		$l(\mathsf{L})=C\setminus \mathsf{L}$.
		Vertices $\mathsf{L}$ are winning if and only if $|\mathsf{L}|$ is odd.
		Losing vertices $L$ have a child vertex
		${\mathsf{L}:f_j}$ for each $f_j\in C\setminus \mathsf{L}$ resulting in $|C\setminus
		\mathsf{L}|/2$ many child vertices. Winning vertices $\mathsf{L}$ have
		the single
		child vertex $\mathsf{L}:e_j$ where the last entry
		$\mathsf{last}(\mathsf{L})$ in $\mathsf{L}$
		is $f_j$. All leafs are losing and are labeled with
		$\emptyset$. The tree has
		height $2k$ and $2(k!)$ vertices.

		\item \emph{Muller condition}: The Zielonka tree for Muller conditions
		$U\subseteq \mathcal{P}(V)$ is structured as follows. Winning
		vertices $s$ have labels
		$l(s)\in U$ and losing vertices $t$ have labels $l(t)\notin U$.
		Winning vertices $s$ have a child vertex $t$ for every maximal
		set $C$ of colors
		among the sets $D$ of colors for which $D\subseteq l(s)$ and
		$D\not\models\varphi$.
		Similarly, losing vertices $t$ have a child vertex $s$ for every
		maximal set $C$ of colors among the sets $D$ of colors for which
		$D\subseteq l(t)$ and $D\models\varphi$.
		The tree has height at most $k$ and branching-width at most $2^k$.

\end{enumerate}
\end{example}
\end{toappendix}

\begin{lemma}\label{lem:ZielonkaTreeSize}
	The height and the branching width of $\mathcal{Z}_\varphi$ are bounded
	by $|C|$ and
	$2^{|C|}$ respectively; the number of vertices in $\mathcal{Z}_\varphi$
	is bounded by $e|C|!$ (where $e$ is Euler's number).
\end{lemma}
\begin{toappendix}
\vspace{15pt}
\noindent\textbf{Proof of Lemma~\ref{lem:ZielonkaTreeSize}:}
\begin{proof}
	The root in $\mathcal{Z}_\varphi$ is labeled with $C$.
	The claim on the height of $\mathcal{Z}_\varphi$ thus follows since
	descending
	a level in $\mathcal{Z}_\varphi$ removes at least one color from vertex
	labels.
	The bound on the branching width of $\mathcal{Z}_\varphi$ is immediate
	since siblings
	are required to have incomparable labels and since the labels of
	child vertices
	are contained in the labels of their parents. This directly yields an
	upper bound
	$2^{|C|^2}$ on the number of vertices in $\mathcal{Z}_\varphi$. The tight
	bound $e |C|!$ is obtained as follows.
	Consider the maximal labeled tree $T^m_C$ such that $l(r)=C$ and for
	every
	$s$ in $T^m_C$, $s$ has $|l(s)|$ children each labeled with
	$l(s)\setminus \{c\}$ for some $c\in l(s)$.
	Let $t(i)$ be the number of vertices of $T^m_C$ for $|C|=i$.
	By the construction of $T^m_C$ we have $t(i+1)=(i+1)\cdot t(i)+1$.
	We prove by induction that
	$t(i)=i! \cdot \displaystyle\sum_{k=0}^i \frac{1}{k!}$.
	This holds for $t(0)=1$.
	Then by induction:
	$$t(i+1)=(i+1)\cdot t(i)+1=
	(i+1)!\cdot \sum_{k=0}^i \frac{1}{k!} + \frac{(i+1)!}{(i+1)!} =
	(i+1)!\cdot \sum_{k=0}^{i+1}\frac{1}{k!}
	$$
	Finally, $\sum_{k=0}^\infty \frac{1}{k!}=e$ and
	$\sum_{k=0}^i \frac{1}{k!}\leq e$ for all $i\geq 0$.
	The Zielonka tree $\mathcal{Z}_\varphi$ over the set of colors $C$ can be seen
	as a tree constructed over the vertices of $T^m_C$ (using each vertex in $T^m_C$
	at most once):
    the root in $\mathcal{Z}_\varphi$ corresponds to $r$,
    and given a vertex $s$ in $\mathcal{Z}_\varphi$
    that already corresponds to some vertex $s'$ in $T^m_C$, a child vertex $t$ of $s$ in
    $\mathcal{Z}_\varphi$ corresponds to an arbitrary
    descendant $t'$ of $s'$ in $T^m_C$ such that $l(t)=l(t')$;
    by maximality of vertex labels in Zielonka trees,
    every two different vertices in $\mathcal{Z}_\varphi$
    correspond to different vertices in $T^m_C$. Thus $\mathcal{Z}_\varphi$
    can be obtained from $T^m_C$ by (possibly) taking shortcuts, and removing
    vertices that are bypassed or evaded by the shortcuts.
	Hence, the claimed bound follows.

	Regarding tightness of the bound, consider a formula $\varphi$ that expresses
	that the number of colors that is visited infinitely often is
	even~\cite{Thomas02}.
	Then the Zielonka tree $\mathcal{Z}_\varphi$ is just $T^m_C$, that is, every vertex from
	$T^m_C$ is required and no shortcuts can be taken.
\end{proof}
\end{toappendix}

\begin{toappendix}

Next we describe how Zielonka trees can be used to characterize winning plays,
following~\cite{DziembowskiJW97}. This result is used in the correctness
proof for the novel parity game
reduction presented in Section~\ref{sec:solvingELgames}.

\begin{definition}[Induced walk and dominating vertex]
	\label{defn:ZielonkaPlay}
    Given an infinite play $\pi=v_0 v_1\ldots$ of $G$,
    a sequence $\rho=t_0 s_1 t_1 s_2\ldots$ of vertices of
    the Zielonka tree is an \emph{induced walk} of $\pi$ if
    \begin{itemize}
    \item the walk starts at the minimal leaf in $\mathcal{Z}_\varphi$, that is, if $t_0=\min_\preceq L_r$,
    \item for each $i$, we have $s_{i+1}=\max_{\leq}\{s\in \mathcal{Z}_\varphi
	\mid s \leq t_i	\wedge \gamma(v_i)\subseteq l(s)\}$,
    that is, $s_{i+1}$ is the lower-most ancestor of $t_i$ that
    contains the set $\gamma(v_i)$ of colors in its label; we say that
    $s_{i+1}$ is the \emph{anchor node} of $t_i$ at $v_i$,
    \item for each $i$, we have that $t_{i+1}$ is some leaf below $s_{i+1}$.
    If $s_{i+1}\in T_\Box$, then we additionally require the following in order
    to cycle fairly through children of winning vertices.
    Let $q=|R(s_{i+1})|$ be the number of child vertices of $s_{i+1}$,
    let $1\leq o\leq q$ be \emph{the} number such that $t_{i}$
    is below the $o$-th child of $s_{i+1}$, and put $j=(o \mod q)+1$;
    then we require that $t_{i+1}$ is some leaf below the $j$-th child of $s_{i+1}$.
    \end{itemize}
	As $\mathcal{Z}_\varphi$ is a tree,
	there is, for each walk $\rho$ that is induced by $\pi$,
	a unique top-most vertex that is visited infinitely often in $\rho$.
	We refer to this as the \emph{dominating vertex} of the walk, denoted
	$d_\rho$.
\end{definition}

\begin{lemma}\label{lem:ZielonkaPlays}
	For all plays $\pi$ and all walks $\rho$ induced by $\pi$,
	every color that is visited infinitely often
	by $\pi$ is contained in $l(d_\rho)$.
	A play $\pi$ is winning for the existential player if
    and only if it induces some walk $\rho$
    such that $d_\rho$ is a winning vertex.
\end{lemma}
\begin{proof}
	For the first claim, let $c$ be a color that is visited infinitely
	often by $\pi=v_0 v_1\ldots$ and let $i$ be a position such that
	$c\in \gamma(v_i)$
	and such that all
	vertices $l_j,t_j$ in $\rho=l_0t_1
	l_1 t_2\ldots$
	with $j\geq i$ are descendants of the dominating vertex $d_\rho$; such a position exists by
	the assumption that $c$ is visited infinitely often by $\pi$
	and by definition of $d_\rho$.
	By definition of $\rho$, we have $\gamma(v_i)\subseteq
	l(t_{i+1})$, so that, in particular, $c\in l(t_{i+1})$.
	By the assumptions on $i$, $t_{i+1}$ is a descendant of $d_\rho$.
	By definition of $\mathcal{Z}_\varphi$, we have $l(t)\subseteq l(s)$ for
	every descendant $t$ of vertex $s$. Thus $c\in l(d)$.

	For the second claim, let $A=\{c\in C\mid \forall i.\exists j>i.\,c\in \gamma(v_j)\}$
	denote the set of colors that are visited infinitely often by $\pi=v_0 v_1\ldots$.
	It suffices to show that there is a walk $\rho$ induced by $\pi$ such that
	$A\models\varphi$ if and only if $l(d_\rho)\models \varphi$.
	Let $\rho$ be the walk that is obtained by resolving both existential
    and universal choices in a fair manner, that is, use the branching mechanism
    that is enforced in induced walks on vertices from $T_\Box$ also 
    for vertices from $T_\bigcirc$.	By the previous item, we have
	$A\subseteq
	l(d_\rho)$.
	For every child vertex $t\in R(d_\rho)$ of $d_\rho$, we have,
	by definition of $\mathcal{Z}_\varphi$, that $l(t)\not\models\varphi$ if
	and only if $l(d_\rho)\models\varphi$. By definition
	of $d_\rho$, we furthermore have
	$A\not\subseteq l(t)$ for all $t\in R(d_\rho)$:
    since $\rho$ visits all child vertices of $d_\rho$ infinitely often,
    there is a position $i$ at which $\rho$ visits $t$
    and such that all colors that are visited in $\pi$ from position $i$
    on are contained in $A$. As $l(d_\rho)\neq l(t)$
    and $\rho$ visits $d_\rho$ infinitely often,
    there is a color $c\in A$ such that $c\notin l(t)$.
	By definition of $\mathcal{Z}_\varphi$, child vertices are picked to have
	maximal labels. This implies that, for every set $D\subseteq l(d_\rho)$
	such
	that $D\cap A=\emptyset$, we have
	$l(d_\rho)\models\varphi$ if and only
	if $l(d_\rho)\setminus D\models\varphi$. Picking $D=l(d_\rho)\setminus
	A$ so that
	$l(d_\rho)\setminus D=A$, we obtain that
	$l(d_\rho)\models \varphi$
	if and only if $A\models\varphi$.
\end{proof}
\end{toappendix}


\section{Solving Emerson-Lei Games}
\label{sec:solvingELgames}

We now show how to extract from the Zielonka tree of an Emerson-Lei
objective a fixpoint characterization of the winning regions of an
Emerson-Lei game.
Solving the game then reduces to computing the fixpoint, yielding a game solving
algorithm that works by fixpoint iteration and hence is directly open to symbolic
implementation. The algorithm is adaptive in the sense that the
structure of
its recursive calls is extracted from the Zielonka tree and hence tailored to
the objective. As a stepping stone towards obtaining our fixpoint
characterization, we first show how Zielonka trees can be used to reduce
Emerson-Lei games to parity games that are structured into tree-like subgames.

Recall that $G=(V,V_\exists,V_\forall,E,\alpha_{\gamma,\varphi})$
is an Emerson-Lei game and that the associated Zielonka tree
is $\mathcal{Z}_\varphi=(T,R,l)$ with set $L$ of leaves and root $r$.
Following~\cite{DziembowskiJW97}, we define the \emph{anchor node} of $v\in V$ and $t\in T$ by $\mathsf{anchor}(v,t)=
\max_{\leq}\{s\in T
	\mid s \leq t	\wedge \gamma(v)\subseteq l(s)\}$; it is the
lower-most ancestor of $t$ that contains $\gamma(v)$ in its label.

\paragraph{A novel reduction to parity games.}

Intuitively, our reduction annotates nodes in $G$
with leaves of $\mathcal{Z}_\varphi$ that act as a memory,
holding information about
the order in which colors have been visited. In the reduced game,
the memory value $t\in L$
is updated according to a move from $v$ to $w$ in $G$ by playing a subgame
along the Zielonka tree. This subgame starts
at the anchor node of $v$ and $t$ and the players in turn
pick child vertices, with the existential player choosing
the branch that is taken at vertices from $T_\bigcirc$ and the
universal player choosing
at vertices from $T_\Box$.\footnote{Players choose from
nodes where they lose, which explains the notation $T_\Box$ and
$T_\bigcirc$.}
Once this subgame reaches a leaf $t'\in L$, the
memory value is updated to $t'$ and another step of $G$ is played.
Due to the tree structure of $\mathcal{Z}_\varphi$
every play in the reduced game (walking through the Zielonka tree in the described way) has a unique
topmost vertex from $T$ that it visits infinitely often; the label of this
vertex corresponds to the set of colors that is visited infinitely often
by the according play of $G$. A parity condition can be used to decide whether this vertex is winning or losing.

Formally, we define the parity game $P_G=(V',V'_\exists,V'_\forall,E',\Omega)$, played over $V'=V\times T$, as follows.
Nodes $(v,t)\in V'$ are owned by the existential player if either
$t\notin L$ and $t\in T_\bigcirc$, or $t\in L$ and $v\in V_\exists$;
all other nodes are owned by the universal player.
Moves and priorities are defined by
\begin{align*}
E'(v,t) &= \begin{cases}
\{v\}\times R(t)&\text{$t\notin L$}\\
E(v)\times \{\mathsf{anchor}(v,t)\}&\text{$t\in L$}
\end{cases}
&
\Omega(v,t)&=\begin{cases}
2\cdot \mathsf{lev}(t) & t\in T_\Box\\
2\cdot \mathsf{lev}(t)+1& t\in T_\bigcirc
\end{cases}
\end{align*}
for $(v,t)\in V'$. We note that $|V'|=|V|\cdot|T|\leq |V|\cdot e|C|!$
by Lemma~\ref{lem:ZielonkaTreeSize}.

\begin{theorem}\label{thm:parityReduction}
For all $v\in V$,
the existential player wins $v$ in the Emerson-Lei game $G$ if and only if
the existential player wins $(v,r)$ in the parity
game $P_G$.
\end{theorem}

\begin{toappendix}
\noindent\textbf{Proof of Theorem~\ref{thm:parityReduction}:}
\begin{proof}
For one direction, let $\sigma$ be a winning strategy for
the existential player in the Emerson-Lei game $G$. We define a 
strategy $\tau:V'_\exists\times M\to V'$ with memory $M$
for the existential player in the parity game $P_G$ as follows.
Each memory value enodes a function $f:V\times T_\bigcirc\to T$
that maps pairs $(v,t)$ to some child vertex $t'\in R(t)$.
The memory is used to memorize past branching choices at nodes
$(v,t)$ such that $t\notin L$ and $t\in T_\bigcirc$, and will be used to ensure
that $\sigma$ cycles fairly through all child vertices of such vertices.
Let $(v,t)\in V'$ be a node such that the existential player wins $v$ in $G$ with 
strategy $\sigma$.
If $t\notin L$ and $t\in T_\bigcirc$, then let $f\in M$
be a memory value as described above; we put $\tau(v,t,f)=(v,t')$, where
$t'=f(t)+1\mod |R(t)|$ is the next node according to memory $f\in M$.
If $t\in L$ and $v\in V_\exists$, then we put $\tau(v,t,f)=(\sigma(v),t')$, where
$t'$ is the anchor node of $v$ and $t$. The memory updating function is defined
as expected. This clearly defines a valid strategy $\tau$.

It remains to show that $\tau$ is a winning strategy for the existential player.
So let $\rho$ be a play of $P_G$ that is compatible with $\tau$.
By construction, $\rho$ induces a play $\pi$ of $G$ that is compatible with $\sigma$.
We have to show that the maximal 
priority that is  visited infinitely
often by $\rho$ is even. As $\rho$ can be seen as a walk trough the Zielonka
tree, this amounts to showing that the dominating vertex in 
$\rho$ belongs to $T_\Box$. We note that $\rho$ is a walk through the Zielonka
tree that is induced by $\pi$.
By Lemma~\ref{lem:ZielonkaPlays}, it suffices to show that the existential
player wins $\pi$. But this is the case since $\pi$ is compatible
with $\sigma$ and $\sigma$ is winning strategy for the existential player.

For the converse direction, let $\tau:V'_\exists\to V'$ 
be a positional winning strategy
for the existential player in the parity game $P_G$.
We define a 
strategy $\sigma:V_\exists\times M\to V$ with memory $M=L$ (that is, memory
values are leafs in the Zielonka tree)
for the existential player in the Emerson-Lei game $G$ as follows.
Let $(v,t)\in V'$ such that $v\in V_\exists$ and
$t\in L$ be a node that the existential player wins with 
strategy $\tau$. Then we have $\tau(v,t)=(w,t')$ where
$w\in E(v)$ and $t'$ is the anchor node of $v$ and $t$;
put $\sigma(v,t)=w$. Memory values in $\sigma$ are updated
according to the Zielonka tree component in $P_G$, that is,
a partial play in $P_G$ that leads from $(v,t)$ to the next node $(w,t')$ 
such that $t'\in L$ updates the memory
for $\sigma$ from $t$ to $t'$.
For a move $(v,w)$ in $G$ that is compatible
with $\sigma$ under memory value $t\in L$, 
update the memory to $t'$, where $t'$ is obtained from the partial play in $P_G$
that leads from $(w,\mathsf{anchor}(v,t))$ to $(w,t')$, where
existential branching at nodes $(v,t)$ with $t\notin L$ is resolved 
using $\tau$, and branching at universal anchor nodes is resolved in a fair manner.
This clearly defines a valid strategy $\sigma$.

It remains to show that $\sigma$ is a winning strategy for the existential player.
So let $\pi$ be a play of $G$ that is compatible with $\sigma$.
By Lemma~\ref{lem:ZielonkaPlays}, it suffices to show that $\pi$ induces some
walk through the Zielonka tree in which the dominating vertex belongs to $T_\Box$.
By construction of $\sigma$, $\pi$ induces a play $\rho$ of $P_G$ that is compatible
with $\tau$. This play $\rho$ is an induced walk through the Zielonka tree.
As $\tau$ is a winning strategy, the maximal priority that is  visited infinitely
often by $\tau$ is even. Thus the dominating vertex in $\rho$ belongs to $T_\Box$.
\end{proof}

\end{toappendix}

This reduction yields a novel indirect method to solve Emerson-Lei games with
$n$ nodes and $k$ colors by solving parity games with $n\cdot ek!$ nodes
and $2k$ priorities; by itself, this reduction does not improve upon
using later appearance records~\cite{HunterD05}.
However, the game $P_G$ consists of subgames of particular tree-like
shapes. The remainder of this section is dedicated to showing how the
special structure of $P_G$ allows for direct
symbolic solution by solving equivalent systems of fixpoint equations
over $V$ (rather than over the exponential-sized set $V'$).

\noindent
\paragraph{Fixpoint equation systems.}
Recall (from e.g.~\cite{BaldanKMP19}) that a hierarchical system
of fixpoint equations is given by equations
$$
X_i =_{\eta_i} f_i(X_1,\ldots,X_k)
$$
for $1\leq i\leq k$, where $\eta_i\in\{\mathsf{GFP},\mathsf{LFP}\}$ and the
$f_i:\mathcal{P}(V)^k\to \mathcal{P}(V)$ are
\emph{monotone} functions, that is, $f_i(A_1,\ldots,A_k)\subseteq
f_i(B_1,\ldots,B_k)$ whenever $A_j\subseteq B_j$ for all $1\leq j\leq k$.
As we aim to use fixpoint equation systems to characterize winning regions of games, it is convenient to define the semantics of equation systems also in terms of games, as proposed in~\cite{BaldanKMP19}. For a system $S$ of $k$ fixpoint equations,
the \emph{fixpoint game} $G_S=(V,V_\exists,V_\forall,E,\Omega)$
is a parity game with sets of nodes $V_\exists=V\times\{1,\ldots,k\}$ and
$V_\forall= \mathcal{P}(V)^k$.
Edges $E$ and priority function $\Omega:V\to\mathbb{N}$
of game nodes are defined,
for $(v,i)\in V_\exists$ and $\bar{A}=(A_1,\ldots,A_k)\in V_\forall$, by
\begin{align*}
E(v,i)&=\{\bar{A}\in V_\forall\mid
v\in f_i(\bar{A})\} & E(\bar{A})&=\{(v,i)\in V_\exists\mid
v\in A_i\}
\end{align*}
and by $\Omega(v,i)=2i-\iota_i$ and $\Omega(\bar{A})=0$,
where $\iota_i=1$ if $\eta_i=\mathsf{LFP}$ and $\iota_i=0$
if $\eta_i=\mathsf{GFP}$.
We say that $v$ is contained in the \emph{solution} of
variable $X_i$ (denoted by $v\in\sem{X_i}$) if and only if the existential
player wins the node $(v,i)$ in $G_S$.
In order to show containment of a node $v$ in the solution of $X_i$, the existential
player thus has to provide a solution $(A_1,\ldots,A_k)\in V_\forall$ for
all variables such that $v\in f_i(A_1,\ldots,A_k)$; the universal player
in turn can challenge a claimed solution $(A_1,\ldots,A_k)$ by picking
some $1\leq i\leq k$ and $v\in A_i$ and moving to $(v,i)$. The game objective
checks whether the dominating equation in a play (that is, the equation
with minimal index among the equations that are evaluated infinitely often in the play) is
a least or a greatest fixpoint equation.

Baldan et al. have shown in~\cite{BaldanKMP19}
that this game characterization is equivalent
to the more traditional
Knaster-Tarski-style definition of the semantics of fixpoint
equation systems in terms of nested fixpoints of the involved functions $f_i$.

To give a flavor of the close connection between fixpoint equation systems and winning
regions in games, we recall that for a given set $V$ of nodes, the \emph{controllable predecessor function}
	$\mathsf{CPre}:2^V\to 2^V$ is defined, for $X\subseteq V$, by
$$
		\mathsf{CPre}(X)=\{v\in V_\exists\mid
		E(v)\cap X\neq\emptyset\}\cup
		\{v\in V_\forall\mid
		E(v)\subseteq X\}.
$$

\begin{example}\label{ex:cpre}
Given a B\"uchi game
$(V,V_\exists,V_\forall,E,\mathsf{Inf}~f)$ with coloring function $\gamma:V\to 2^{\{f\}}$,
the winning region of the existential player is the solution of the equation system
\begin{align*}
X_1 & =_{\mathsf{GFP}} X_2 & X_2 & =_{\mathsf{LFP}} (f\cap\mathsf{CPre}(X_1))\cup (\overline{f}\cap\mathsf{CPre}(X_2))
\end{align*}
where $f=\{v\in V\mid \gamma(v)=\{f\}\}$ and $\overline{f}=V\setminus f$.
\end{example}


Our upcoming fixpoint characterization of winning regions in Emerson-Lei games
uses the following notation that relates game nodes with anchor nodes in
the Zielonka tree.

\begin{definition}\label{defn:fts}
	For a set $D\subseteq C$ of colors, and ${\bowtie}\in \{\subseteq,
\not\subseteq\}$ we put
$\gamma^{-1}_{\bowtie D}=\{v\in V\mid \gamma(v)\, \bowtie
\,D\}$.	For
	$s,t\in T$ such that $s< t$ (that is,
	$s$ is an ancestor of $t$ in $\mathcal{Z}_\varphi$),
	we define
	\begin{align*}
		\mathsf{anc}^s_{t}=\gamma^{-1}_{\subseteq l(s)}\cap
		\gamma^{-1}_{\not\subseteq l(s_t)}
	\end{align*}
	where $s_t$ is \emph{the} child vertex of $s$ that
	leads to $t$; we also put
	$\mathsf{anc}^t_{t}=\gamma^{-1}_{\subseteq l(t)}$.
\end{definition}
Note that for fixed $t\in T$ and $v\in V$, there is a unique $s\in T$
such that $s\leq t$ and $v\in \mathsf{anc}^s_{t}$ (possibly, $s=t$); this $s$ is the
anchor node of $t$ at $v$.

Next, we present our
fixpoint characterization of winning in Emerson-Lei games,
noting that it closely follows the definition of $P_G$.

\begin{definition}[Emerson-Lei equation system]
	\label{defn:fp}
We define the system $S_\varphi$ of fixpoint equations for the objective $\varphi$ by putting
	\begin{align*}
		X_s &=_{\eta_s}
		\begin{cases}
			\textstyle\bigcup_{t\in R(s)}X_{t} &
			R(s)\neq\emptyset,s\in T_\bigcirc\\
			\textstyle\bigcap_{t\in R(s)}X_{t} &
			R(s)\neq\emptyset,s\in T_\Box\\
			\textstyle\bigcup_{s'\leq s} \left(\mathsf{anc}^{s'}_{s} \cap\mathsf{CPre}(X_{s'})\right) &
			R(s)=\emptyset
		\end{cases}
	\end{align*}
	for $s\in T$.
    For every $t\in T$, we use $X_t$ to refer to the variable $X_i$
    where $i$ is the index of $t$
    according to $\preceq$ and similarly for $\eta_t$.
	Furthermore, $\eta_t=\mathsf{GFP}$ if $t\in T_\Box$
	and $\eta_t=\mathsf{LFP}$ if $t\in T_\bigcirc$.
\end{definition}

\begin{example}\label{example:eqSys}
Instantiating Definition~\ref{defn:fp} to the B\"uchi objective $\varphi=\mathsf{Inf}~f$
yields exactly the equation system given in Example~\ref{ex:cpre}.
Revisiting the objectives from Example~\ref{example:zielonkaTrees}, we
obtain the following fixpoint characterizations 
(further examples can be found in the appendix).
	\begin{enumerate}
		\item \emph{Generalized B\"uchi condition:}
		\begin{align*}
			X_{s_0} &=_\mathsf{GFP}\textstyle\bigcap_{1\leq i\leq k}X_{s_i} &
			X_{s_i} &=_\mathsf{LFP}
			(\mathsf{anc}^{s_0}_{s_i}\cap \mathsf{CPre}(X_{s_0})) \cup
			(\mathsf{anc}^{s_i}_{s_i}\cap \mathsf{CPre}(X_{s_i}))
		\end{align*}
		where $\mathsf{anc}^{s_0}_{s_i}=\gamma^{-1}_{\subseteq C}
			\cap \gamma^{-1}_{\not\subseteq C\setminus\{f_i\}}=\{v\in V\mid f_i\in \gamma(v)\}$ and
			$\mathsf{anc}^{s_i}_{s_i}=\gamma^{-1}_{\subseteq C\setminus\{f_i\}}$.\medskip

		\item \emph{Streett condition}:
		\begin{align*}
			X_{\mathsf{L}}&=_{\eta_{\mathsf{L}}}
			\begin{cases}
				\textstyle \bigcap_{g_j\notin {\mathsf{L}}}X_{{\mathsf{L}}:g_j}  & |{\mathsf{L}}| \text{ even},
				|{\mathsf{L}}|<2k\\
				X_{{\mathsf{L}}:r_j}& |{\mathsf{L}}| \text{ odd},\mathsf{last}({\mathsf{L}})=g_j\\
				(\mathsf{anc}^{[]}_{{\mathsf{L}}}\cap \mathsf{CPre}(X_{[]})) \cup \ldots \cup
				(\mathsf{anc}^{{\mathsf{L}}}_{{\mathsf{L}}}\cap \mathsf{CPre}(X_{{\mathsf{L}}})) & |{\mathsf{L}}|=2k
			\end{cases}
		\end{align*}
		where $\eta_{\mathsf{L}}=\mathsf{GFP}$ if $|{\mathsf{L}}|$ is even and
		$\eta_{\mathsf{L}}=\mathsf{LFP}$ if $|{\mathsf{L}}|$ is odd. Here,
			$\mathsf{anc}^{\mathsf{K}}_{{\mathsf{L}}}=\gamma^{-1}_{\subseteq C\setminus
				\mathsf{K}}\cap\gamma^{-1}_{\not\subseteq
				C\setminus I}$ for $\mathsf{K}\neq {\mathsf{L}}$ and $I=\mathsf{K}_{\mathsf{L}}$, and
			$\mathsf{anc}^{{\mathsf{L}}}_{{\mathsf{L}}}=\gamma^{-1}_{\subseteq\emptyset}$,
		both for ${\mathsf{L}}$ such that $|{\mathsf{L}}|=2k$.
    \item
    The equation system associated to the Zielonka tree for the complex
    objective $\varphi_{EL}$ from Example~\ref{example:zielonkaTrees}.3
    is as follows, where we use a formula over the colors
    to denote the set of vertices whose label satisfies the formula.
    For example, $b \wedge \neg d$ corresponds
    to vertices whose set of colors contains $b$ but does not contain
    $d$.

\begin{scriptsize}
    \begin{align*}
X_{1}&=_{\mathsf{LFP}} X_{2}\cup X_{3} \qquad X_{2} =_{\mathsf{GFP}} X_{4}\cap X_{5}\qquad
X_{3}=_{\mathsf{GFP}} X_{6}\qquad\,\, X_{5} =_{\mathsf{LFP}} X_{7}\qquad
X_{7} =_{\mathsf{GFP}} X_{8}\\
X_{4} &=_{\mathsf{LFP}} (\neg c\land \neg d
\cap\mathsf{Cpre}(X_4))\cup(c\land \neg
d\cap\mathsf{Cpre}(X_{2}))\cup(d\cap\mathsf{Cpre}(X_{1}))\\
X_{6} &=_{\mathsf{LFP}} (\neg a\land\neg c\cap\mathsf{Cpre}(X_6))\cup
(\neg a\land c\cap\mathsf{Cpre}(X_3))\cup(a\cap\mathsf{Cpre}(X_{1}))\\
X_{8} &=_{\mathsf{LFP}}
(\neg a\land\neg b \land \neg c\land \neg
d\cap\mathsf{Cpre}(X_{8}))\cup(\neg a\land \neg b \land
c\land \neg d \cap\mathsf{Cpre}(X_{7}))\,\cup\\
&\quad\quad\,\,\,(a \land \neg b \land \neg d
\cap\mathsf{Cpre}(X_{5}))\cup(b \land
\neg d\cap\mathsf{Cpre}(X_{2}))\cup(d\cap\mathsf{Cpre}(X_{1})),
\end{align*}
\end{scriptsize}

	\end{enumerate}
\end{example}

\begin{toappendix}

\begin{example}
We expand on Example~\ref{example:eqSys} and give further instantiations
of the generic equation system from Definition~\ref{defn:fp} to
games with parity and Rabin objectives.
\begin{enumerate}
		\item \emph{Parity condition}:
		\begin{align*}
			X_{s_i} &=_{\eta_{s_i}}\begin{cases}
				X_{s_{i-1}} & i>1\\
				(\mathsf{anc}^{s_k}_{s_1}\cap \mathsf{CPre}(X_{s_k})) \cup \ldots \cup
				(\mathsf{anc}^{s_1}_{s_1}\cap \mathsf{CPre}(X_{s_1})) & i=1
			\end{cases}
		\end{align*}
		where $\eta_{s_i}=\mathsf{GFP}$ if $i$ is even and
		$\eta_{s_i}=\mathsf{LFP}$ otherwise;
		we have
		\begin{align*}
			\mathsf{anc}^{s_i}_{s_1}&=\gamma^{-1}_{\subseteq\{p_1,\ldots,p_i\}}\cap
			\gamma^{-1}_{\not\subseteq\{p_1,\ldots,p_{i-1}\}}\tag{$i>1$}\\
			\mathsf{anc}^{s_1}_{s_1}&=\gamma^{-1}_{\subseteq\{p_1\}}
		\end{align*}
		assuming that every node has exactly one of the colors
		$p_1,\ldots,p_k$.
		\item \emph{Rabin condition}:
		\begin{align*}
			X_{\mathsf{L}} & =_{\eta_{{\mathsf{L}}}}
			\begin{cases}
				\textstyle\bigcup_{f_j\notin {\mathsf{L}}}X_{{\mathsf{L}}:f_j} & |{\mathsf{L}}|\text{ even},
				|{\mathsf{L}}|<2k\\
				X_{{\mathsf{L}}:e_j} & |{\mathsf{L}}|\text{ odd},\mathsf{last}({\mathsf{L}})=f_j\\
				(\mathsf{anc}^{[]}_{{\mathsf{L}}}\cap \mathsf{CPre}(X_{[]})) \cup \ldots \cup
				(\mathsf{anc}^{{\mathsf{L}}}_{{\mathsf{L}}}\cap \mathsf{CPre}(X_{{\mathsf{L}}})) & |{\mathsf{L}}|=2k
			\end{cases}
		\end{align*}
		where $\eta_{\mathsf{L}}=\mathsf{GFP}$ if $|{\mathsf{L}}|$ is odd and
		$\eta_{\mathsf{L}}=\mathsf{LFP}$ if $|{\mathsf{L}}|$ is even. We have
		\begin{align*}
			\mathsf{anc}^{{\mathsf{K}}}_{{\mathsf{L}}}&=\gamma^{-1}_{\subseteq C\setminus
				\mathsf{K}}\cap\gamma^{-1}_{\not\subseteq
				C\setminus I}\tag{${\mathsf{K}}\neq {\mathsf{L}}$, $I={\mathsf{K}}_{\mathsf{L}}$}\\
			\mathsf{anc}^{{\mathsf{L}}}_{{\mathsf{L}}}&=\gamma^{-1}_{\subseteq\emptyset}
		\end{align*}
		for ${\mathsf{L}}$ such that $|{\mathsf{L}}|=2k$.

\end{enumerate}
\end{example}
\end{toappendix}

\begin{theorem}\label{thm:correctness}
	Referring to the equation system
	from Definition~\ref{defn:fp} and recalling that $r$ is the root
	of the Zielonka tree $\mathcal{Z}_\varphi$, the solution of the
	variable $X_r$ is the winning region of the existential player in the Emerson-Lei game $G$.
\end{theorem}
By Theorem~\ref{thm:parityReduction}, it suffices to
mutually transform
winning strategies in $P_G$ and the fixpoint game $G_{S_\varphi}$
for the equation system $S_\varphi$ from Definition~\ref{defn:fp}.

\begin{toappendix}
\noindent\textbf{Proof of Theorem~\ref{thm:correctness}:}

\begin{proof}
By Theorem~\ref{thm:parityReduction}, it suffices to
show that the solution of the variable $X_r$ in the equation system
	from Definition~\ref{defn:fp} is the winning region
 of the existential player in the parity game $P_G$.
 Due to the similarities between $P_G$ and the equation system $S_\varphi$,
 this is relatively straightforward.
The formal proof is by transforming winning strategies $\sigma$ for $P_G$ into
winning strategies $\tau$ for the fixpoint game $G_{S_\varphi}$, and vice versa.

For the backward direction,
let $V$ be the set of nodes in $G$ and let $\tau$ be a positional winning strategy for the
        existential player
        in the fixpoint game $G_{S_\varphi}$ with which the existential player wins
        every node in its winning region; such a strategy exists since
        $G_{S_\varphi}$ is a parity game.
        We define a strategy $\sigma:V'_\exists\to V'$
        for the existential player in $P_G$. So let $(w,t)\in V'_\exists$ 
        be a node such that
        the existential player wins from $(v,t)$ in $G_{S_\varphi}$ with strategy $\tau$.
        We have $\tau(v,t)=\{U_{u_1},\ldots,U_{u_n}\}$.
        If $t\notin L$ and $t\in T_\bigcirc$, then some set
        $U_{u_i}$ such that $u_i\in R(t)$ contains $v$. Pick
        some such $u_i$ and put $\sigma(v,t)=(v,u_i)$.
        If $t\in L$ and $v\in V_\exists$ then
        let $s$ be the anchor node of $v$ and $t$.
        Then $v\in\mathsf{anc}^s_t\cap \mathsf{CPre}(U_{s})$
        by definition of $S_\varphi$. Hence 
        there is $w\in E(v)\cap U_s$. Pick some such $w$ and put
        $\sigma(v,t)=(w,s)$. This defines a valid strategy for $P_G$.

        It remains to show that $\sigma$ is a winning strategy for the existential
        player.
        So let $\pi$ be a play in $P_G$ that is compatible with $\sigma$.
        By construction, $\pi$ induces a play $\rho$ of $G_{S_\varphi}$
        that is by construction compatible with $\tau$.
        As $\tau$ is a winning strategy, the top-most vertex $s_\pi$ in $T$
        that is visited infinitely often by $\rho$ is a
        winning vertex (from $T_\Box$). By definition of the priority function $\Omega$ of
        $P_G$, the existential player wins $\pi$.

	For the forward direction of the claim,
    let $\sigma:V'_\exists\to V'$ be a positional winning
	strategy for the existential player in $P_G$
	such that the existential player wins every play that starts
	in its winning region and is compliant with $\sigma$.
	We build a strategy $\tau$ in the fixpoint game $G_{S_\varphi}$
	as follows.
	Let $(v,t)$ be a node in $P_G$
	that the existential player wins with strategy $\sigma$.
    If $t\notin L$ and $t\in T_\bigcirc$, then
    we have $s(v,t)=(v,t')$ for some $t'\in R(t)$.
    Put $\tau(v,t)=(\emptyset,\ldots,\emptyset,U_{t'},\emptyset,\ldots,\emptyset)$
    where $U_{t'}=\{v\}$.
    If $t\notin L$ and $t\in T_\Box$, then we
    put $\tau(v,t)=(\emptyset,\ldots,\emptyset,U_{t_1},\ldots, U_{t_|R(t)|},\emptyset,\ldots,\emptyset)$, where $R(t)=\{t_1,\ldots, t_{|R(t)|}\}$ and $U_{t_i}=\{v\}$.
    If $t\in L$, then let $Z(v,t)=s(v,t)$ if $v\in V_\exists$
    and $Z(v,t)=E(v)$ if $v\in V_\forall$. Put $\tau(v,t)=(\emptyset,\ldots,\emptyset,U_{t'},\emptyset,\ldots,\emptyset)$ where $t'$ is the anchor
    node of $v$ and $t$ and where $U_{t'}=Z(v,t)$.
    Clearly, $\tau$ is a valid strategy for the existential player.

	It remains to show that $\tau$ is a winning strategy for the existential player.
	So let $\pi$ be a play in the fixpoint game $G_{S_\varphi}$ that is compliant with
	$\tau$. Then $\pi$ induces a play $\rho=v_0v_1\ldots$ of $P_G$ that is
	compatible with $\sigma$.
	As $\sigma$ is a winning strategy, the maximal priority
	that is visited infinitely often
	by $\rho$ is even, so the topmost vertex in $T$ that
	is visited infinitely often by $\pi$ and $\rho$ belongs to $T_\Box$. Thus the   existential player wins $\rho$.
\end{proof}
\end{toappendix}

Given the fixpoint characterization of winning regions
in Emerson-Lei games with objective $\varphi$ in
Definition~\ref{defn:fp}, we obtain a fixpoint iteration algorithm that computes
the solution of Emerson-Lei games. The algorithm
is by nature open to symbolic implementation. The main function is recursive,
taking as input
one vertex $s\in T$ of the Zielonka tree $\mathcal{Z}_\varphi$ and
a list $l$ of subsets of the set $V$ of nodes, and returns
a subset of $V$ as result.
For calls $\textsc{Solve}(s,ls)$, we
require that the argument list $ls$ contains exactly one
subset $X_{s'}$ of $V$ for each ancestor $s'$ of $s$ in the Zielonka
tree (with $s'<s$).

\begin{algorithm}
\caption{$\textsc{Solve}(s,ls)$}\label{alg:exists}
\begin{algorithmic}
\State \textbf{if} $s\in T_\bigcirc$\textbf{ then }$X_s \gets \emptyset$
\textbf{ else }$X_s \gets V$\Comment{Initialize variable $X_s$ for
lfp/gfp}
\State $W \gets V\setminus X_s$
\While{$X_s \neq W$}\Comment{Compute fixpoint}
\State $W \gets X_s$
\If{$R(s)\neq\emptyset$}\Comment{Case: $s$ is not a leaf in $\mathcal{Z}_\varphi$}
    \For{$t\in R(s)$}
    \State $U \gets \textsc{Solve}(t,ls:W)$\Comment{Recursively solve
    for $t$}
    \State
    \textbf{if} $s\in T_\bigcirc$\textbf{ then }$X_s\gets X_s\cup U$\\
    \qquad\qquad\qquad\qquad\,\,\,\,\,\textbf{ else }$X_s\gets X_s\cap U$
    \EndFor\\
\quad\,\,\textbf{else}\Comment{Case: $s$ is a leaf in $\mathcal{Z}_\varphi$}
    \State $Y\gets\emptyset$
    \For{$t\leq s$}
    \State $U \gets \mathsf{anc}^t_s\cap
    \mathsf{CPre}((ls:W)(t))$\Comment{Compute one-step attraction w.r.t.
    $s$}
    \State $Y\gets Y\cup U$
    \EndFor
    \State $X_s\gets Y$ 
\EndIf
\EndWhile\\
\Return $X_s$\Comment{Return stabilized set $X_s$ as result}
\end{algorithmic}
\end{algorithm}

\begin{lemma}\label{lem:ELcorrectness}
For all $v\in V$, we have
$v\in\sem{X_r}$ if and only if $v\in\textsc{Solve}(r,[])$.
\end{lemma}
\begin{proof}[Sketch]
The algorithm computes the solution of the equation system by
standard Kleene-approximation
for nested least and greatest fixpoints.
\end{proof}
\begin{toappendix}
\noindent\textbf{Proof of Lemma~\ref{lem:ELcorrectness}:}
\begin{proof}
By the equivalence of fixpoint games and the traditional Knaster-Tarski semantics
of fixpoint equation systems (as shown in~\cite{BaldanKMP19}),
it suffices to show that $\textsc{Solve}(r,[])$
computes the Knaster-Tarski semantics of $X_r$; the latter defines
the solution $\sem{X_i}$ of an equation $X_i=_{\mathsf{GFP}} f(X_1,\ldots,X_n)$
as the union of all postfixpoints of $f$ with respect to
the variable $X_i$, and similarly,
the solution of an equation $X_i=_{\mathsf{LFP}} g(X_1,\ldots,X_n)$
as the intersection of all prefixpoints of $g$. Since $V$ is a finite
set, Kleenes fixpoint theorem states that
these solutions can also be approximated by repeated application of $f$
and $g$, starting with arguments $V$ and $\emptyset$, respectively.
That is, we have $\mathsf{GFP} X.\, f(X)=f^{|V|}(V)$ and
$\mathsf{LFP} X.\, g(X)=g^{|V|}(\emptyset)$, where the repeated application
of a function $h:\mathcal{P}(V)\to\mathcal{P}(V)$ to a set
$W\subseteq V$ is defined inductively
by $h^0(W)=W$ and $h^{i+1}(W)=h(h^i(W))$. The function
$\textsc{Solve}$ performs exactly
this approximation of $\sem{X_r}$: $\textsc{Solve}(s,l)$
computes an extremal fixpoint by iteratively applying the right-hand side
$f_s$ of the equation for $X_s$ to the initial set $V$ or $\emptyset$, depending
on whether $s\in T_\bigcirc$ or $s\in T_\Box$, and
taking the argument sets from $l$ as fixed inputs for the respective
variables. The recursive calls in the algorithm
 correspond to descending steps in the Zielonka tree $\mathcal{Z}_\varphi$, following
Definition~\ref{defn:fp}.
\end{proof}
\end{toappendix}
\begin{lemma}\label{lem:ELcomplexity}
Given an Emerson-Lei game $(V,V_\exists,V_\forall,E,\alpha_{\gamma,\varphi})$ with set of colors $C$
and induced Zielonka tree $\mathcal{Z}_\varphi$, the solution $\sem{X_r}$
of the equation system $S_\varphi$ from Definition~\ref{defn:fp}
can be computed in time $\mathcal{O}(|\mathcal{Z}_\varphi|\cdot |E|\cdot|V|^k)$,
where $k\leq|C|$ denotes the height of $\mathcal{Z}_\varphi$.
\end{lemma}

\begin{toappendix}
\noindent\textbf{Proof of Lemma~\ref{lem:ELcomplexity}:}
\begin{proof}
By Lemma~\ref{lem:ELcorrectness}, it suffices to show
that for all vertices $s$ in $\mathcal{Z}_\varphi$
and all lists $l$ of subsets of $V$ that contain exactly one subset $X_t$ for each ancestor
$t<s$ of $s$,
the runtime of $\textsc{Solve}(s,l)$ is bounded by
$m_s\cdot |E|\cdot |V|^{d_s}$, where
$m_s$ and $d_s$ denote the number of vertices and the depth, respectively,
of the subtree of $\mathcal{Z}_\varphi$ that is rooted at $s$.
The claim of the lemma follows by picking $s=r$.
The proof is by induction on $d_s$.

If $s\in T_\bigcirc$, then $X_s$ is initialized to $V$ and, by monotonicity of
all involved operations (in particular, $\mathsf{CPre}$ is a monotone function),
every iteration of the
while loop in $\textsc{Solve}(s,l)$ removes at least one element from $X_s$,
until, eventually, $X_s=W$ and the loop terminates.
If $s\in T_\Box$, then $X_s$ is initialized to $\emptyset$ and each iteration of the
while loop
(except the final iteration) adds at least one element of $V$ to $X_s$.
Hence the number of iterations of the while loop for fixed $s$ and $l$
is always is bounded by $|V|$.

\begin{itemize}
\item In the base case, $s$ is a leaf node and we have $m_s=d_s=1$.
Since $R(s)=\emptyset$, the runtime of a single iteration
of the while loop is dominated by the computation of $\mathsf{CPre}(X_t)$
for all $t\leq s$; this can be implemented to run in time $|E|$, yielding the
claimed bound.
\item In the inductive case, we have $R(s)\neq \emptyset$ so that
the runtime of a single iteration of the while loop is bounded by the sum of the runtimes
of $\textsc{Solve}(t,l:X_s)$ over all $t\in R(s)$.
Denoting the runtime of the respective
algorithm when applied to arguments $s$ and $l$ by $\sigma(s,l)$, we hence have
$\sigma(s,l)\leq |V|\cdot \Sigma_{t\in R(s)} \sigma(t,l:X_s)$.
By the inductive hypothesis, we have $\sigma(t,l:X_s)\leq m_t\cdot|E|\cdot |V|^{d_t}$
for all child vertices $t\in R(s)$ of $s$. Also, $m_s=\Sigma_{t\in R(s)} m_t +1 $
and $d_t<d_s$ for all $t\in R(s)$. Hence
\begin{align*}
\sigma(s,l)& \leq |V|\cdot \Sigma_{t\in R(s)} \sigma(t,l:X_s)\\
&\leq |V|\cdot \Sigma_{t\in R(s)} m_t\cdot|E|\cdot |V|^{d_t}\\
&\leq |V|\cdot \Sigma_{t\in R(s)} m_t\cdot|E|\cdot |V|^{d_s-1}\\
&=|V|\cdot |E|\cdot |V|^{d_s -1}\cdot \Sigma_{t\in R(s)} m_t\\
&\leq m_s\cdot |E|\cdot |V|^{d_s},
\end{align*}
as required.
\end{itemize}
\end{proof}
\end{toappendix}

Combining Theorem~\ref{thm:correctness} with Lemmas~\ref{lem:ZielonkaTreeSize},
~\ref{lem:ELcorrectness} and~\ref{lem:ELcomplexity} we obtain
\begin{corollary}\label{cor:complexity}
Solving Emerson-Lei games with $n$ nodes, $m$ edges and $k$ colors
can be implemented symbolically to run in time $\mathcal{O}(k!\cdot m\cdot n^k)$;
the resulting strategies require memory at most $k!$.
\end{corollary}

\begin{remark}~\label{remark:stratextr}
Strategy extraction works as follows.
The algorithm computes a set $\sem{X_t}$ for each Zielonka tree
vertex $t\in\mathcal{Z}_\varphi$.
Furthermore it yields, for each non-leaf vertex $s\in T_\bigcirc$
and each $v\in\sem{X_s}$, a single child vertex $\mathsf{choice}(v,s)\in R(s)$
of $s$ such that $v\in\sem{X_{\mathsf{choice}(v,s)}}$.
The algorithm also yields, for each leaf vertex $t$
and each $v\in V_\exists\cap \sem{X_t}$, a single game move $\mathsf{move}(v,t)$.
All these choices together constitute a winning strategy for existential
player in the parity game $P_G$.
We define a strategy for the Emerson-Lei game that uses leaves
of the Zielonka tree as memory values, following
the ideas used in the construction of $P_G$; the strategy moves,
from a node $v\in V_\exists$
and having memory content $m$, to the node $\mathsf{move}(v,m)$.
As initial memory value we pick some leaf of $\mathcal{Z}_\varphi$
that $\mathsf{choice}$ associates with the initial node in $G$.
To update memory value $m$ according to visiting game node $v$,
we first take the anchor node $s$ of $m$ and $v$. Then we pick the next memory value $m$
to be some leaf below $s$ that can be reached by talking the choices $\mathsf{choice}(v,s')$
for every vertex $s'\in T_\circ$ passed along the way from $s$ to the leaf;
if $s\in T_\Box$, then we additionally require the following:
let $q=|R(s)|$, let $o$ be the number such that $m$ is a leaf below the $o$-th child of $s$,
and put $j=o+1\mod q$;
then we require that $m'$ is a leaf below the $j$-th child of $s$.
By the correctness of the algorithm, the constructed strategy
is a winning strategy.

Dziembowski et al. have shown
that winning strategies can be extracted by using a walk
through the Zielonka tree that requires memory only for the branching at winning vertices~\cite{DziembowskiJW97}.
This yields, for instance, memoryless strategies for games with Rabin objectives, for which
branching in the associated Zielonka trees takes place at losing nodes.
Adapting the strategy extraction in our setting to this more economic method
is straight-forward but notation-heavy, so we omit a more precise analysis of strategy size here.

\end{remark}

Our algorithm hence can be implemented to run in time $2^{\mathcal{O}(k \log n)}$
for games with $n$ nodes and $k\leq n$ colors, improving upon
the bound $2^{\mathcal{O}(n^2)}$ stated in~\cite{HunterD05}, where the authors only
consider the case that every game node has a distinct color, implying $n=k$.
We note that the later appearance record
construction used in~\cite{HunterD05} is known to be hard to represent symbolically.
Our fixpoint characterization generalizes previously known algorithms
for e.g. parity games~\cite{BruseFL14}, and Streett and Rabin games~\cite{PitermanP06},
recovering previously known bounds on worst-case running time of fixpoint
iteration algorithms for these types of games.

We note that in the case of parity objectives, our algorithm is not quasipolynomial. However, there are quasipolynomial methods for solving
nested fixpoints~\cite{HausmannSchroeder21,ArnoldNP21} (with the latter
being open to symbolic implementation); in the case of parity
objectives, these more involved algorithms can be used in place of
fixpoint iteration to solve our equation system and
recover the quasipolynomial bound.
The precise complexity of using quasipolynomial methods for
solving fixpoint equation systems beyond parity conditions
is subject to ongoing research.


\section{Synthesis for Safety and Emerson-Lei LTL}
\label{sec:synt}


In this section we present an application of the results from Section~\ref{sec:solvingELgames}.
We introduce the safety and Emerson-Lei fragment of LTL and show that
synthesis for this fragment can be reasoned about symbolically.
The idea for safety and Emerson-Lei LTL synthesis is twofold: first,
consider only the safety part and create a symbolic game capturing its
satisfaction.
Second, play the game adding the Emerson-Lei part as a winning
condition.
Finally we use the results from the previous sections to solve the game symbolically.

\subsection{Safety LTL and Symbolic Safety
Automata}\label{sec:safetyLTLandSafetyAutomata}
We start by defining safety LTL, symbolic safety automata, and recalling known results about those.

\begin{definition}[Safety LTL]
Given a non-empty set $\AP$ of atomic propositions, the general
syntax for LTL formulas is as follows:
$\varphi := \top \mid \bot \mid p \mid \neg \varphi \mid \varphi_1
\land \varphi_2 \mid \varphi_1 \lor \varphi_2 \mid
X \varphi \mid \varphi_1 U \varphi_2 \mid F \varphi \mid G
\varphi\qquad\ p\in\AP$.

%

We define the satisfaction relation $\models$ for a formula $\varphi$ and its language $\Lang(\varphi)$ as usual.
We also define the release operator $R$ as $\varphi_1 R \varphi_2 := \neg(\neg \varphi_1 U \neg \varphi_2)$.


An LTL formula is said to be a \emph{safety formula} if it is in negative
normal form (i.e. all negations are pushed to atomic propositions) and
only uses $X, R, G$ as temporal operators (i.e. no $U$ or $F$ are
allowed
).
\end{definition}

It is a safety formula in the sense that every word that does not
satisfy the formula has a finite prefix that already falsifies the
formula.
In other words, such a formula is satisfied as long as ``bad states''
are avoided forever.


\begin{definition}[Symbolic Safety Automata]
A symbolic safety automaton is a tuple $\A = (2^{\AP}, V, T, \theta_0)$
where $V$ is a set of variables, $T(V,V',\AP)$ is the transition
assertion, and $\theta_0(V)$ is the initialization assertion.
A run of $\A$ on the word $w \in (2^{\AP})^\omega$  is a sequence
$\rho = s_0 s_1 \dots$ where the $s_i \in 2^V$ are
variable assignments such that
\begin{inparaenum}
	\item $s_0 \models \theta_0$, and
	\item for all $i \geq 0$, $(s_i, s_{i+1}, w(i)) \models T$.
\end{inparaenum}
A word $w$ is in $\Lang(\A)$ if and only if there is an infinite run of $\A$ on $w$.
\end{definition}

Kupferman and Vardi show how to convert a safety LTL formula into an equivalent deterministic symbolic safety automaton \cite{kupferman2001model}.

\begin{lemma}
A safety LTL formula $\varphi$ can be translated to a deterministic symbolic safety automaton $\DSA_\symb$ accepting the same language,
with $|\DSA_\symb| = 2^{|\varphi|}$.
\end{lemma}

The idea is to first convert $\varphi$ to a (non-symbolic) non-deterministic safety automaton $\NSA_\varphi$, which is of size exponential of the size of the formula, and then symbolically determinize $\NSA_\varphi$ by a standard subset construction to obtain $\DSA_\symb$.
Note that while the size of $\DSA_\symb$ is only exponential in the size of the formula, its state space would be double exponential when fully expanded.

\begin{example}\label{example:safety}
Let $\varphi = G(b \lor c) \land G(a \rightarrow b \lor XXb)$ be a safety LTL formula over $\AP = \{a,b,c\}$.
An execution satisfying $\varphi$ must have at least one of $b$ or $c$
at every step, moreover every $a$ sees a $b$ present at the same step
or two steps afterwards.

As an intermediate towards building the equivalent $\DSA_\symb$, we
first present below a corresponding non-deterministic safety automaton $\NSA_\varphi$.
\vspace{-7pt}
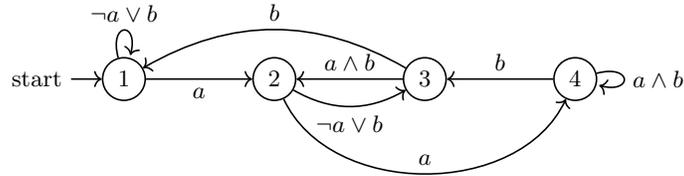
\begin{footnotesize}
\begin{center}
\tikzset{every state/.style={minimum size=15pt}}
  \begin{tikzpicture}[
		auto,
    node distance=2cm,
    semithick
    ]
     \node[state,initial left] (1) {$1$};
     \node[state] (2) [right of=1] {$2$};
     \node[state] (3) [right of=2] {$3$};
     \node[state] (4) [right of=3] {$4$};
     \path[->] (1) edge [loop above] node [above] {$\neg a \lor b$} (1);
     \path[->] (1) edge node [below] {$a$} (2);
     \path[->] (2) edge [bend right] node [below] {$\neg a \lor b$} (3);
     \path[->] (2) edge [bend right=65] node [above] {$a$} (4);
     \path[->] (4) edge [loop right] node [right] {$a \land b$} (4);
     \path[->] (4) edge node [above] {$b$} (3);
     \path[->] (3) edge [bend right] node [above] {$b$} (1);
     \path[->] (3) edge node [above] {$a \land b$} (2);
  \end{tikzpicture}
\end{center}

\end{footnotesize}
\vspace{-10pt}
For the sake of presentation, we use Boolean combinations of $\AP$ in
transitions instead of labeling them with elements of $2^\AP$, with the
intended meaning that $s \xrightarrow{\psi} s' = \{s
\xrightarrow{C} s' \mid C \in 2^\AP,~C \models \psi\}$.
We also omit the $G(b \lor c)$ part of the formula in the construction.
One can simply append $\dots \land (b \lor c)$ to every transition of $\NSA_\varphi$ to get back the original formula.
Intuitively state 1 correspond to not seeing an $a$, state 2 means that a $b$ must be seen at the next step, state 3 means that there must be a $b$ now, and state 4 that $b$ is needed now and next as well.

Then the symbolic safety automaton is $\DSA_\symb = (2^\AP, V, T, \theta_0)$ with:
\begin{itemize}
\item $V = \{v_1,v_2,v_3,v_4\}$ are the variables corresponding to the four states of $\NSA_\varphi$,
\item $\theta_0 = v_1 \land \neg v_2 \land \neg v_3 \land \neg v_4$ asserts that only the state $v_1$ is initial,
\item The transition assertion is
$T = \,(v'_1 \leftrightarrow (v_1 \land (\neg a \lor b)) \lor (v_3 \land b))\,\land\\
(v'_2 \leftrightarrow (v_1 \land a) \lor (v_3 \land (a \land b)))\,\land
(v'_3 \leftrightarrow (v_2 \land (\neg a \lor b)) \lor (v_4 \land b))\,\land\\
(v'_4 \leftrightarrow (v_2 \land a) \lor (v_4 \land (a \land b)))\,\land
(v_1 \lor v_2 \lor v_3 \lor v_4)$.
\end{itemize}
Determinizing $\NSA_\varphi$ enumeratively would give an automaton with 9 states (see Example~\ref{ex:gameArena}).
\end{example}

\begin{remark}
Restricting attention to safety LTL enables the two advantages mentioned above with respect to determinization.
First, subset construction suffices (as observed also in \cite{abs-2008-06790}), avoiding the more complex B\"uchi determinization.
Second, this construction, due to its simplicity, can be implemented symbolically.
Interestingly, recent implementations of the synthesis from LTL$_f$
\cite{abs-2008-06790} or from safety LTL \cite{ZhuTLPV17} have used
indirect approaches for obtaining deterministic automata.
For example, by translating LTL to first order logic and applying the
tool MONA to the results \cite{ZhuTLPV17,abs-2008-06790} or
concentrating on minimization of deterministic automata
\cite{TabakovV05}.
The direct construction is similar to approaches
used for checking universality of nondeterministic finite automata
\cite{TabakovV05} or SAT-based bounded model checking
\cite{ArmoniEFKV05}.
We are not aware of uses of this direct implementation of the subset
construction in reactive synthesis.
The worst case complexity of this part is doubly-exponential, which,
just like for LTL and LTL$_f$, cannot be avoided \cite{VardiS85,abs-2211-14913}.
\end{remark}

\subsection{Symbolic Games}
We use \emph{symbolic game structures} to represent a certain
class of games.
Formally, a \emph{symbolic game structure}
${\cal G} = \langle {\cal
	V}, {\cal X}, {\cal Y},
\theta_\sys, \rho_\sys, \varphi \rangle$ consists of:
\begin{itemize}
	\item[$\bullet$] ${\cal V} = \{v_1,\ldots, v_n\}$ : A finite set of
	typed
	{\em variables} over finite domains. Without loss of
	generality, we assume they are all Boolean.  A node $s$ is an
	valuation of ${\cal V}$, assigning to each variable $v_i\in {\cal
	V}$ a value $s[v_i]\in \{0,1\}$.
	Let $\Sigma$ be the set of nodes.

	We extend the evaluation function $s[\cdot]$ to Boolean
	expressions over ${\cal V}$ in the usual way.
	An \emph{assertion} is a
	Boolean formula over ${\cal V}$.
	A node $s$ satisfies an assertion
	$\varphi$ denoted $s\models \varphi$, if
	$s[\varphi]={\bf true}$.
	We say that $s$ is a $\varphi$-node if $s\models
	\varphi$.
	\item[$\bullet$] ${\cal X} \subseteq {\cal V}$ is a
	set of \emph{input variables}.
	These are variables controlled by the universal player.
	Let
	$\Sigma_{\cal X}$ denote the possible valuations to
	variables in ${\cal X}$.
	\item[$\bullet$] ${\cal Y} = {\cal V} \setminus
	{\cal X}$ is a set of \emph{output variables}.
	These are variables controlled by the existential player.
	Let $\Sigma_{\cal Y}$ denote the possible valuations to
	variables in ${\cal Y}$.
	\item[$\bullet$]
	$\theta_\sys({\cal X},{\cal Y})$ is an assertion characterizing the
	initial
	condition.
	\item[$\bullet$]
	$\rho_\sys({\cal V},{\cal X}',{\cal Y}')$ is the
	transition relation.
	This is an assertion relating a node $s\in
	\Sigma$ and an input value $s_{\cal X} \in \Sigma_{\cal X}$ to an output value
	$s_{\cal Y} \in \Sigma_{\cal Y}$ by referring to primed and unprimed
	copies of ${\cal V}$.
	The transition relation $\rho_\sys$ identifies a valuation $s_{\cal Y}\in
	\Sigma_{\cal Y}$ as a \emph{possible output} in node $s$ reading input
	$s_{\cal X}$ if $(s,(s_{\cal X},s_{\cal Y})) \models \rho_\sys$, where $s$
	is the assignment to variables in $\cal V$ and
	$s_{\cal X}$ and $s_{\cal Y}$ are the assignment to variables
	in $\cal V'$ induced by $(s_{\cal X},s_{\cal Y})\in\Sigma$.
	\item[$\bullet$] $\varphi$ is the winning condition, given by an LTL
	formula.
\end{itemize}

\noindent
For two nodes $s$ and $s'$ of ${\cal G}$, $s'$ is a \emph{successor} of
$s$ if $(s,s') \models \rho_\sys$.

A symbolic game structure ${\cal G}$ defines an arena $A_{\cal G}$,
where $V_\forall=\Sigma$, $V_\exists = \Sigma \times \Sigma_{\cal X}$,
and $E$ is defined as follows:
$$
E= \{(s,(s,s_{\cal X})) ~|~ s\in \Sigma \mbox{ and } s_{\cal X}\in
\Sigma_{\cal X} \} \cup \{((s,s_{\cal X}),(s_{\cal
X},s_{\cal Y})) ~|~ (s,(s_{\cal X},s_{\cal Y}))\models \rho_\sys\}.$$
When reasoning about symbolic game structures we ignore the
intermediate visits to $V_\exists$.
Indeed, they add no information as they can be deduced from the nodes
in $V_\forall$ preceding and following them.
%
Thus, a play $\pi=\,s_0s_1\ldots$ is \emph{winning for the existential player}
if $\sigma$ is infinite and satisfies $\varphi$.
Otherwise, $\sigma$ is \emph{winning for the universal player}.

The notion of strategy and winning region is trivially generalized
from games to symbolic game structures.
When needed, we treat $W_\exists$ (the set of nodes winning for the existential player) as an assertion.
We define winning in the \emph{entire} game structure by incorporating
the initial assertion:
a game structure ${\cal G}$ is said to
be \emph{won} by the existential player, if for all $s_{\cal X} \in
\Sigma_{\cal X}$ there exists $s_{\cal Y} \in \Sigma_{\cal Y}$ such that
$(s_{\cal X},s_{\cal Y})\models \theta_\sys \land W_\exists$.

\subsection{Realizability and Synthesis}
Let $\varphi$ be an LTL formula over input and output
variables $I$ and $O$, controlled by \emph{the environment} and
\emph{the system}, respectively (the universal and  the
existential player, respectively).

The reactive synthesis problem asks whether there is a strategy for the
system of the form $\strat:
(2^I)^+ \to 2^O$ such that for all sequences $x_0 x_1 \dots \in
(2^I)^\omega$ we have:
\[(x_0 \cup \strat(x_0)) (x_1 \cup \strat(x_0 x_1)) \dots \models
\varphi\]
If there is such a strategy we say that $\varphi$ is \emph{realizable}
\cite{PnueliR89}.

%

Equivalently,
$\varphi$ is \emph{realizable} 
if the system is winning in the symbolic game
${\cal G}_\varphi =\langle I \cup O,
I, O, tt,tt,\varphi\rangle$ with $I$ for input variables ${\cal X}$ and
$O$ for output ${\cal Y}$.

\begin{theorem}{\cite{PnueliR89}}
	Given an LTL formula $\varphi$, the realizability of $\varphi$ can be
	determined in doubly exponential time. The problem is
	2EXPTIME-complete.
\end{theorem}

The game ${\cal G}_\varphi$ above uses neither the initial condition
nor the system transition.
Conversely, consider a symbolic game ${\cal G}=\langle
{\cal V},{\cal X},{\cal
	Y},\theta_\sys,\rho_\sys,\varphi\rangle$:

\begin{theorem}{\cite{BloemJPPS12}}
	The system wins in ${\cal G}$ iff $\varphi_{_{\cal G}}=
	\theta_\sys \wedge G \rho_\sys  \wedge
	\varphi$ is
	realizable.%
\footnote{
	Technically, $\rho_\sys$ contains primed variables and
	is not an LTL formula.
	This can be easily handled by using the next operator $X$.
	We thus ignore this issue.
}%
\footnote{
	We note that Bloem et al. consider more general games,
	where the environment also has an initial assertion
	and a transition relation.
	Our games are obtained from
	theirs by setting the initial assertion and the
	transition relation of the environment to true.
}
	\label{thm:realizable}
\end{theorem}

\subsection{Safety and Emerson-Lei Synthesis}\label{sec:el-synt}

We now define the class of LTL formulas that are supported by
our technique and show how to construct appropriate games
capturing their realizability problem.

For all $\psi \in \mathbb{B}(\AP)$, let $\Inf\,\psi :=
GF\psi$
and $\Fin\,\psi := FG\neg\psi = \neg \Inf\,\psi$.
The \emph{Emerson-Lei fragment} of LTL consists of all formulas that
are positive
Boolean combinations of $\Inf\,\psi$ and $\Fin\,\psi$ for
all Boolean
formulas $\psi$ over atomic propositions.
The satisfaction of such formulas depends only on the set of
letters
(truth assignments to propositions) appearing infinitely often
in a
word.

\begin{remark}
The Emerson-Lei fragment easily accommodates various liveness properties that cannot 
be encoded in smaller fragments such as GR[1].
One prominent example for this is the property of \emph{stability} (as encoded by LTL
formulas of the shape $\mathsf{FG}~p$), which appears
frequently in usage of synthesis for robotics and control (see, e.g.,
the work of Ehlers~\cite{Ehlers11a} and Ozay~\cite{LiuOTM13}), and commonly is
approximated in GR[1] but cannot be captured exactly in the game context.
Another important example is \emph{strong fairness} (as encoded by LTL
formulas of the shape $\bigwedge_{i} (\mathsf{GF}~r_i\to\mathsf{GF}~g_i)$) which
allows to capture the exact relation between cause and effect.
Particularly, in GR[1] only if \emph{all} ``resources'' are available infinitely
often there is an obligation on the system to supply \emph{all} its ``guarantees''.
In contrast, strong fairness allows to connect particular resources to
particular supplied guarantees. Ongoing studies on fairness assumptions that arise from the
abstraction of continuous state spaces to discrete state spaces~\cite{LiuOTM13,MajumdarS23}
provide further examples
of fairness assumptions that can be expressed in EL but not in GR[1].
Emerson-Lei liveness allows free combination of all properties mentioned above and more.
\end{remark}

\begin{definition}\label{defn:elsfragment}
	The \emph{Safety and Emerson-Lei fragment} is the set of
	formulas of the form
	$\varphi = \varphi_{\mathrm{safety}} \land
	\varphi_{\mathrm{EL}}$,
	where $\varphi_{\mathrm{safety}}$ is a safety formula and
	$\varphi_{\mathrm{EL}}$ is in the Emerson-Lei fragment.
\end{definition}

We assume a partition $\AP
= I \uplus O$ where $I$ is a set of \emph{input propositions} and $O$ a set of
\emph{output propositions}, both non-empty.
%
Let $\varphi = \varphi_{\mathrm{safety}} \land
\varphi_{\mathrm{EL}}$ be a safety and Emerson-Lei formula over $\AP$, and let $\DSA_\symb =
(2^\AP, V, T, \theta_0)$ be the symbolic deterministic safety automaton
associated to $\varphi_{\mathrm{safety}}$.
We construct $G_\varphi = \langle V \uplus \AP, I, O \uplus V, \theta_0, T,
\varphi_{\mathrm{EL}} \rangle$, thus ${\cal X}=I$ and ${\cal Y}=O
\uplus V$.

%
%

\begin{example}\label{ex:gameArena}
Let $\varphi_{\mathrm{safety}} = G(b \lor c) \land G(a \rightarrow b \lor XXb)$, our running safety example from Example~\ref{example:safety} with its associated symbolic deterministic automaton.
Partition $\AP$ into $I = \{a\}$ and $O = \{b,c\}$.
We depict the arena of the game $G_\varphi$ (independent of the formula
$\varphi_{\mathrm{EL}}$ that is yet to be defined) in
Figure~\ref{fig:example:arena}.

\begin{figure}
\tikzset{every state/.style={minimum size=15pt}}
\begin{center}
  \begin{tikzpicture}[
		auto,
    node distance=1.5cm,
    semithick
    ]
     \node[state,initial left,rectangle] (1) {$v_1$};
     \node[draw,circle] (r1) [right of=1] {$$};
     \node[draw,rectangle] (12) [above right of=r1] {$v_1,v_2$};
     \node[draw,rectangle] (2) [below right of=r1] {$v_2$};
     \node[draw,rectangle] (13) [above of=12] {$v_1,v_3$};
     \node[draw,circle] (r12) [right of=12] {$$};
     \node[draw,rectangle] (1234) [above right of=r12] {$v_1,v_2,v_3,v_4$};
     \node[draw,rectangle] (24) [below right of=r12] {$v_2,v_4$};
     \node[draw,circle] (r24) [right of=24] {$$};
     \node[draw,rectangle] (3) [right of=2] {$v_3$};
     \node[draw,rectangle] (4) [below right of=24] {$v_4$};
     \node[draw,circle] (r2) [right of=4] {$$};
     \node[draw,rectangle] (34) [above right of=r24] {$v_3,v_4$};
     \node[draw,circle] (r13) [below left of=13] {$$};
     \node[draw,circle] (r1234) [below right of=1234] {$$};

     \path[->] (1) edge [loop below] node [below] {$\neg a;\ast$} (1);
     \path[->] (1) edge node [above] {$a$} (r1);
     \path[->] (r1) edge node [right] {$b$} (12);
     \path[->] (r1) edge node [below left] {$\neg b$} (2);
     \path[->] (12) edge node [right] {$\neg a;\ast$} (13);
     \path[->] (12) edge node [above] {$a$} (r12);
     \path[->] (r12) edge node [above] {$b$} (1234);
     \path[->] (r12) edge node [below left] {$\neg b$} (24);
     \path[->] (2) edge node [above] {$\neg a;\ast$} (3);
     \path[->] (2) edge [bend right] node [below] {$a$} (r2);
     \path[->] (r2) edge node [above] {$\neg b$} (4);
     \path[->] (r2) edge node [right] {$b$} (34);
     \path[->] (13) edge [bend right] node [above left] {$\neg a;\ast$} (1);
     \path[->] (13) edge node [above left] {$a$} (r13);
     \path[->] (r13) edge node [above right] {$b$} (12);
     \path[->] (r13) edge node [below right] {$\neg b$} (2);
     \path[->] (1234) edge node [below] {$\neg a;\ast$} (13);
     \path[->] (1234) edge [bend right] node [below] {$a$} (r1234);
     \path[->] (r1234) edge node [right] {$b$} (1234);
     \path[->] (r1234) edge node [below] {$\neg b$} (24);
     \path[->] (24) edge node [right] {$\neg a;\ast$} (3);
     \path[->] (24) edge node [below] {$a$} (r24);
     \path[->] (r24) edge node [left] {$\neg b$} (4);
     \path[->] (r24) edge node [left] {$b$} (34);
     \path[->] (3) edge node [right] {$a;b$} (12);
     \path[->] (3) edge [bend left=45] node [below left] {$\neg a;b$} (1);
     \path[->] (4) edge node [right] {$a;b$} (34);
     \path[->] (4) edge node [below] {$\neg a;b$} (3);
     \path[->] (34) edge node [above] {$a;b$} (1234);
     \path[->] (34) edge [bend right=40] node [right] {$\neg a;b$} (13);
  \end{tikzpicture}
\caption{Game arena for $G_\varphi$}
\end{center}
\label{fig:example:arena}
\end{figure}
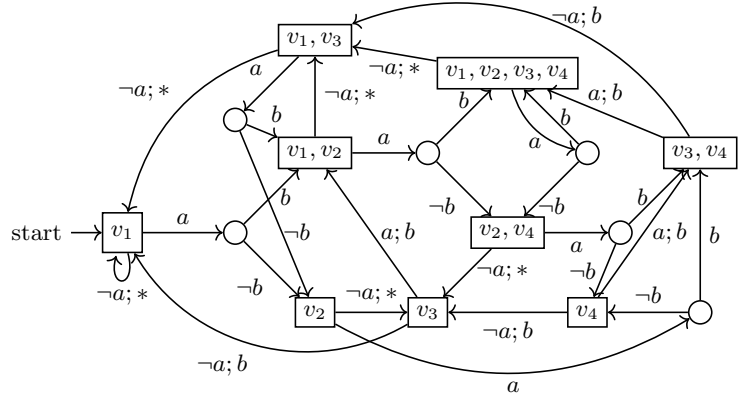

To keep the illustration readable and keep it from getting too large, a few modifications to the formal arena definition have been made.
First, $c$ labels on edges have been omitted: every transition labeled
with $b$ represent two transitions with sets $\{b\}$ and $\{b,c\}$,
while transitions labeled with $\neg b$ stand for a single transition
with set $\{c\}$ (due to the $G(b \lor c)$ requirement forbidding
$\emptyset$).
Similarly, existential nodes have been omitted when all choices for the existential player lead to the same destination.
Instead, the universal and existential moves have been combined in one
transition: $a;\ast$ for an $a$ followed by some existential move, and
$a;b$ for when an $a$ requires the existential player to play $b$ (with
or without $c$, as above).
Finally, states are only labeled with variables from $V$ and not $\AP$,
the latter is used to label edges instead.
For a fully state-based labeling arena, states would have to store the last move, leading to various duplicate states.

Note that this game arena is given only for illustration purposes, as we want to solve the symbolic game without explicitly enumerating all its states and transitions like here.
\end{example}

\begin{lemma}\label{thm:realizabilitygames}
	The system wins $G_\varphi$ if and only if $\varphi$ is realizable.
\end{lemma}
\begin{toappendix}
\textbf{Proof of Theorem~\ref{thm:realizabilitygames}:}

\begin{proof}
	Using Theorem~\ref{thm:realizable}, the system wins $G_\varphi$
	if and only if $\theta_0 \land G(T) \land \varphi_{\mathrm{EL}}$ is
	realizable. We show that the latter is equivalent to $\varphi$ being
	realizable.

	\fbox{$\Rightarrow$} Suppose $\theta_0 \land G(T) \land
	\varphi_{\mathrm{EL}}$ is realizable, and let $\strat: (2^I)^\ast \to
	2^{O \uplus V}$ be a strategy that realizes it.
	Let $\overline{\strat}$ be defined by dropping the $V$ component in
	the output
	of $\strat$: $\overline{\strat}(w) = \strat(w){\downarrow}_O$.
	We claim that $\overline{\strat}$ is a winning strategy for $\varphi =
	\varphi_{\mathrm{safety}} \land \varphi_{\mathrm{EL}}$.

	Let $i_0 i_1 \dots$ be a sequence of inputs and $(o_k, v_k) =
	\strat(i_0 \dots i_k)$ for all $k \geq 0$ be the corresponding
	sequence of outputs and variables according to $\strat$ .
	By definition, $(i_0 \cup o_0 \cup v_0) (i_1 \cup o_1 \cup v_1) \dots
	$ satisfies $\theta_0 \land G(T) \land \varphi_{\mathrm{EL}}$, and
	thus satisfies both $\theta_0 \land G(T)$ and $\varphi_{\mathrm{EL}}$.

	From the former, it follows that $v_0$ satisfies $\theta_0$ and that
	$(v_k, v_{k+1}, i_k \cup o_k)$ satisfies $T$ for all $k \geq 0$.
	This entails that $\rho = v_0 v_1 \dots$ is a valid run of
	$\DSA_\symb$ on $w = (i_0 \cup o_0) (i_1 \cup o_1) \dots$, and
	by the results of
	Section~\ref{sec:safetyLTLandSafetyAutomata} we have $w \models
	\varphi_{\mathrm{safety}}$.

	From the latter, we deduce that $w$ also satisfies
	$\varphi_{\mathrm{EL}}$, as this formula is over $\AP$ only so the
	$V$ component is irrelevant to its satisfaction.
	Therefore, $\overline{\strat}$ is indeed a winning strategy for
	$\varphi$.

	\fbox{$\Leftarrow$} Suppose $\varphi$ is realizable and let
	$\overline{\strat}:
	(2^I)^\ast \to 2^O$ be a winning strategy.
	We define $\strat: (2^I)^\ast \to 2^{O \uplus V}$ as follows:
	Let $i_0 i_1 \dots$ be a sequence of inputs and $o_0 o_1 \dots$ the
	corresponding sequence of outputs according to $\overline{\strat}$.
	We have that $(i_0 \cup o_0) \dots$ satisfies $\varphi$, so again by
	Section~\ref{sec:safetyLTLandSafetyAutomata} it also belongs to the
	language of $\DSA_\symb$.
	Thus there is a run $\rho = v_0 v_1 \dots$ of $\DSA_\symb$ over this
	sequence, which is uniquely determined by the sequence of inputs and
	$\overline{\strat}$.
	We put $\strat(i_0 \dots i_k) = o_k \cup v_k$ for all $k \geq 0$.

	Let $w = (i_0 \cup o_0 \cup v_0) \dots$ be a sequence compatible with
	$\strat$.
	Then as stated before we have that $v_0 v_1 \dots$ is a run of
	$\DSA_\symb$, which means that $v_0 \models \theta_0$ and that $(v_k,
	v_{k+1}, i_k \cup o_k) \models T$ for all $k \geq 0$, and therefore
	means that $w$ satisfies both $\theta_0$ and $G(T)$.
	Moreover, as $(i_0 \cup o_0) \dots$ satisfies
	$\varphi_{\mathrm{EL}}$, then so does $w$ because once again the
	satisfaction of $\varphi_{\mathrm{EL}}$ does not depend on the
	sequence of $v_k$.
	We then have that $w$ satisfies $\theta_0 \land G(T) \land
	\varphi_{\mathrm{EL}}$, so $\strat'$ is a strategy for its
	realizability.
\end{proof}

\end{toappendix}



Next we detail how to solve the symbolic game $G_\varphi$ by using the result from Section~\ref{sec:solvingELgames}.

\begin{lemma}\label{lem:symbolicEL}
	Given a symbolic game $G = \langle {\cal V}, {\cal X}, {\cal Y},
	\theta_\sys, \rho_\sys, \varphi \rangle$ such that $\varphi$ is an
	Emerson-Lei formula with set of colors
    \vspace{-5pt}
    \begin{align*}
	C=\{\psi\in\mathbb{B}(\mathsf{AP})\mid \psi\text{ is a subformula of }\varphi\},
	\end{align*}
	the
	winning region $W_\exists$ of $G$ is characterized
	by the equation system from Definition~\ref{defn:fp},
	using the assertion
\vspace{-5pt}
	\begin{align*}
		\mathsf{CPre}(S)=\forall s_{\mathcal{X}} \in \Sigma_{\mathcal{X}}
		.\,\exists s_{\mathcal{Y}} \in \Sigma_{\mathcal{Y}}.\,
		S'\land (v,s_{\mathcal{X}},s_{\mathcal{Y}})\models\rho_\sys.
	\end{align*}
\end{lemma}
The proof of this lemma is by straightforward adaptation of the proof
of Theorem~\ref{thm:correctness}
to the symbolic setting, following the relation between symbolic game structures
and game arenas described above.

Finally, this gives us a procedure to solve the synthesis problem for safety and Emerson-Lei LTL.

\begin{theorem}
The realizability of a formula $\varphi = \varphi_\mathsf{safety}\land\varphi_{EL}$
of the Safety and Emerson-Lei fragment of LTL
can be checked in time $2^{\mathcal{O}({m\cdot\log m\cdot 2^{n}})}$,
where $n=|\varphi_\mathsf{safety}|$ and $m=|\varphi_{EL}|$.
Realizable formulas can be realized by systems of size at most $2^{2^n}\cdot e\cdot m!$.
\end{theorem}
\begin{proof}
Using the construction described in this section,
we obtain the symbolic game $G_{\varphi}$ of
size $q=2^{2^n}$ with winning condition $\varphi_{EL}$, using
at most $m$ colors;
by Theorem~\ref{thm:realizabilitygames}, this
game characterizes realizibility of the formula.
Using the results from the previous section, $G_{\varphi}$
can be solved in time
$\mathcal{O}(m!\cdot q^2 \cdot q^m)\in \mathcal{O}(2^{m \log m}\cdot 2^{(m+2)2^n})\in 2^{\mathcal{O}({m\cdot\log m\cdot 2^{n}})}$, resulting
in winning strategies with memory at most $e\cdot m!$.
\end{proof}
Both the automata determinization and the game solving can be implemented symbolically.

\begin{example}

To illustrate the overall synthesis method, we consider the
game that is obtained by combining the game arena
$G_{\varphi_{\mathsf{safety}}}$ from
Example~\ref{ex:gameArena} with the winning objective $\varphi_{EL}=
(\mathsf{Inf}~a\to
\mathsf{Inf}~b)\wedge(\mathsf{Fin}~a\vee \mathsf{Fin}(b \wedge c))\land
\mathsf{Inf}~c$
from Example~\ref{example:zielonkaTrees}.3, where we instantiate the
label $d$ to nodes satisfying $b\wedge c$ thus creating a game-specific
dependency between the colors.
Solving this game amounts to solving the equation system
shown in Example~\ref{example:eqSys}.3.
However, with the interpretation of $d=b\wedge c$, some of the
conditions become simpler. For example, $\neg a \wedge \neg b \wedge
\neg c \wedge \neg d$ becomes $\neg a \wedge \neg b \wedge \neg c$ and
$b \wedge \neg d$ becomes $b \wedge \neg c$.
It turns out that
the system player wins the node $v_1$. Intuitively, the system can
play $\{c\}$ whenever possible and thereby guarantee satisfaction of $\varphi_{EL}$.
We extract this strategy from the computed
solution of the equation system in Example~\ref{example:zielonkaTrees}.3
as described in Remark~\ref{remark:stratextr}.
E.g. for partial runs $\pi$ that end in $v_1$ and for which the last
leaf vertex in the induced walk $\rho_\pi$ through $\mathcal{Z}_\varphi$
is the vertex $8$, the system
can react by playing $\{b\}$, $\{c\}$, or even $\{b,c\}$
whenever the environment plays $\emptyset$.
The move $\{b\}$ continues the induced walk $\rho_\pi$
through vertex $2$ to the leaf vertex $5$;
similarly, the move $\{b,c\}$ continues $\rho_\pi$ through the vertex $1$
to the leaf vertex $6$.
The strategy construction gives precedence to the choice that leads through
the lowest vertex in the Zielonka tree,
which in this case means picking the move $\{c\}$ that
continues $\rho_\pi$ through the vertex $7$ to the leaf $8$.
Proceeding similarly for all other combinations of game nodes and vertices
in the Zielonka tree, one
obtains a strategy $\sigma$ for the system that always outputs
singleton letters, giving precedence to $\{c\}$ whenever possible.
To see that $\sigma$ is a winning strategy, let $\pi$ be a play that is
compatible with $\sigma$.
If $\pi$ eventually loops at $v_1$ forever, then $s_\pi$ is the existential vertex $7$
and the existential player wins the play since it satisfies both $\mathsf{Fin}~a$
and $\mathsf{Inf}~c$.
Any other play $\pi$
satisfies $\mathsf{Inf}~a$, $\mathsf{Inf}~b$ and $\mathsf{Inf}~c$ since
all cycles that are compatible with $\sigma$ (excluding the loop at $v_1$)
contain at least one $a$-edge, at least one $b$-edge and also
at least one $c$-edge that is prescribed by the strategy $\sigma$.
For these plays, $\rho_\pi$ eventually reaches
the vertex $2$.
Since the system always plays singleton letters (so that
$\pi$ in particular satisfies $\mathsf{Fin}(b\land c)$),
the vertex $1$ is not visited again by $\rho_\pi$,
once vertex $2$ has been reached. Hence the dominating vertex for such
plays is $s_\pi=2$, an existential vertex.
\end{example}

\subsection{Synthesis Extensions and Optimizations}
There are well known approaches to extend the
expressiveness of GR[1]
by including deterministic automata in the safety part
of the game and
referring to their states in the liveness part
\cite{BloemJPPS12}.
The same applies to our approach.
For example, by adding one bit of memory we can include
``eventually''
and ``globally'' formulas in the winning condition
(though ``globally'' conditions can be included in the
safety part).
Past LTL \cite{LichtensteinPZ85} as part of the winning condition can be
handled in the same way in that it is incorporated for GR[1]
\cite{BloemJPPS12}.
Similarly, the Boolean state formulas appearing as part of the EL
condition can be replaced by formulas allowing one usage of the next
operator, as in \cite{RamanPFK15,Ehlers11a}.
The generalization to handle transition-based EL games rather than
state-based EL games is straight-forward.

As the formulas we consider are conjunctions,
optimizations can be
applied to both conjuncts independently.
This subsumes, for example, analyzing the winning region in a safety
game prior to the
full analysis \cite{KuglerS09,BloemJPPS12,BansalGSLVZ22},
reductions in the size of nondeterministic automata
\cite{Duret-LutzRCRAS22}, or symbolic
minimization of deterministic automata \cite{DAntoniV14}.%
\footnote{%
	Notice that explicit minimization as done, e.g., in
	\cite{kupferman2001model} would
	require to explicitly construct the potentially doubly exponential
	deterministic automaton, nullifying the entire effort to keep all
	analysis symbolic.%
}



\section{Conclusions and Future Work}
\label{sec:conc}

We provide a symbolic algorithm to solve games with Emerson-Lei winning
conditions.
Our solution is based on an encoding of the Zielonka tree of the
winning condition in a system of fixpoint equations.
In case of known winning conditions, our algorithm recovers known
algorithms and complexity results.
As an application of this algorithm,
we suggest an expressive fragment of LTL for which
realizability can be reasoned about symbolically.
Formulas in our fragment are conjunctions between an
LTL safety formula and an Emerson-Lei liveness
condition.
This fragment is more general than, e.g., GR[1].

In the future, we believe
that analysis of the Emerson-Lei part can reduce the size of Zielonka trees (and
thus the symbolic algorithm). This can be done either through analysis
and simplification of the LTL formula, e.g., \cite{JohnJBK21},
by means of alternating-cycle decomposition~\cite{CasaresCL22,CasaresDMRS22},
or by analyzing the semantic meaning of colors. We
would also like to implement the proposed overall synthesis method.

\clearpage

\bibliographystyle{splncs04}
\bibliography{lib}

\providecommand{\noopsort}[1]{}
\begin{thebibliography}{10}
\providecommand{\url}[1]{\texttt{#1}}
\providecommand{\urlprefix}{URL }
\providecommand{\doi}[1]{https://doi.org/#1}

\bibitem{ArmoniEFKV05}
Armoni, R., Egorov, S., Fraer, R., Korchemny, D., Vardi, M.Y.: Efficient {LTL}
  compilation for sat-based model checking. In: International Conference on
  Computer-Aided Design. pp. 877--884. {IEEE} Computer Society (2005).
  \doi{10.1109/ICCAD.2005.1560185},
  \url{https://doi.org/10.1109/ICCAD.2005.1560185}

\bibitem{ArnoldNP21}
Arnold, A., Niwinski, D., Parys, P.: A quasi-polynomial black-box algorithm for
  fixed point evaluation. In: Baier, C., Goubault{-}Larrecq, J. (eds.) 29th
  {EACSL} Annual Conference on Computer Science Logic, {CSL} 2021, January
  25-28, 2021, Ljubljana, Slovenia (Virtual Conference). LIPIcs, vol.~183, pp.
  9:1--9:23. Schloss Dagstuhl - Leibniz-Zentrum f{\"{u}}r Informatik (2021).
  \doi{10.4230/LIPIcs.CSL.2021.9},
  \url{https://doi.org/10.4230/LIPIcs.CSL.2021.9}

\bibitem{abs-2211-14913}
Artale, A., Geatti, L., Gigante, N., Mazzullo, A., Montanari, A.: Complexity of
  safety and cosafety fragments of linear temporal logic. CoRR
  \textbf{abs/2211.14913} (2022). \doi{10.48550/arXiv.2211.14913},
  \url{https://doi.org/10.48550/arXiv.2211.14913}

\bibitem{BaldanKMP19}
Baldan, P., K{\"{o}}nig, B., Mika{-}Michalski, C., Padoan, T.: Fixpoint games
  on continuous lattices. Proc. {ACM} Program. Lang.  \textbf{3}({POPL}),
  26:1--26:29 (2019). \doi{10.1145/3290339}

\bibitem{BansalGSLVZ22}
Bansal, S., Giacomo, G.D., Stasio, A.D., Li, Y., Vardi, M.Y., Zhu, S.:
  Compositional safety {LTL} synthesis. In: 14th International Conference on
  Verified Software, Theories, Tools and Experiments. Lecture Notes in Computer
  Science, vol. 13800, pp. 1--19. Springer (2022).
  \doi{10.1007/978-3-031-25803-9\_1},
  \url{https://doi.org/10.1007/978-3-031-25803-9\_1}

\bibitem{BhatiaMKV11}
Bhatia, A., Maly, M.R., Kavraki, L.E., Vardi, M.Y.: Motion planning with
  complex goals. {IEEE} Robotics Autom. Mag.  \textbf{18}(3),  55--64 (2011).
  \doi{10.1109/MRA.2011.942115}

\bibitem{BloemJPPS12}
Bloem, R., Jobstmann, B., Piterman, N., Pnueli, A., Sa'ar, Y.: Synthesis of
  reactive(1) designs. J. Comput. Syst. Sci.  \textbf{78}(3),  911--938 (2012).
  \doi{10.1016/j.jcss.2011.08.007}

\bibitem{BruseFL14}
Bruse, F., Falk, M., Lange, M.: The fixpoint-iteration algorithm for parity
  games. In: International Symposium on Games, Automata, Logics and Formal
  Verification, GandALF 2014. {EPTCS}, vol.~161, pp. 116--130 (2014).
  \doi{10.4204/EPTCS.161.12}

\bibitem{CamachoM19}
Camacho, A., McIlraith, S.A.: Learning interpretable models expressed in linear
  temporal logic. In: Twenty-Ninth International Conference on Automated
  Planning and Scheduling. pp. 621--630. {AAAI} Press (2019).
  \doi{10.1609/icaps.v29i1.3529}

\bibitem{CamachoTMBM17}
Camacho, A., Triantafillou, E., Muise, C.J., Baier, J.A., McIlraith, S.A.:
  Non-deterministic planning with temporally extended goals: {LTL} over finite
  and infinite traces. In: Thirty-First {AAAI} Conference on Artificial
  Intelligence. pp. 3716--3724. {AAAI} Press (2017).
  \doi{10.1609/aaai.v31i1.11058}

\bibitem{CasaresCL22}
Casares, A., Colcombet, T., Lehtinen, K.: On the size of good-for-games rabin
  automata and its link with the memory in muller games. In: Bojanczyk, M.,
  Merelli, E., Woodruff, D.P. (eds.) International Colloquium on Automata,
  Languages, and Programming, {ICALP} 2022. LIPIcs, vol.~229, pp.
  117:1--117:20. Schloss Dagstuhl - Leibniz-Zentrum f{\"{u}}r Informatik
  (2022). \doi{10.4230/LIPIcs.ICALP.2022.117}

\bibitem{CasaresDMRS22}
Casares, A., Duret{-}Lutz, A., Meyer, K.J., Renkin, F., Sickert, S.: Practical
  applications of the alternating cycle decomposition. In: Fisman, D., Rosu, G.
  (eds.) Tools and Algorithms for the Construction and Analysis of Systems -
  28th International Conference, {TACAS} 2022, Held as Part of the European
  Joint Conferences on Theory and Practice of Software, {ETAPS} 2022, Munich,
  Germany, April 2-7, 2022, Proceedings, Part {II}. Lecture Notes in Computer
  Science, vol. 13244, pp. 99--117. Springer (2022).
  \doi{10.1007/978-3-030-99527-0\_6},
  \url{https://doi.org/10.1007/978-3-030-99527-0\_6}

\bibitem{Church63}
Church, A.: Logic, arithmetic, and automata. In: International Congress of
  Mathematicians. Institut Mittag-Leffler, Sweden (1963)

\bibitem{DAntoniV14}
D'Antoni, L., Veanes, M.: Minimization of symbolic automata. In: Symposium on
  Principles of Programming Languages (POPL). pp. 541--554. {ACM} (2014).
  \doi{10.1145/2535838.2535849}, \url{https://doi.org/10.1145/2535838.2535849}

\bibitem{Duret-LutzRCRAS22}
Duret{-}Lutz, A., Renault, E., Colange, M., Renkin, F., Aisse, A.G.,
  Schlehuber{-}Caissier, P., Medioni, T., Martin, A., Dubois, J., Gillard, C.,
  Lauko, H.: From spot 2.0 to spot 2.10: What's new? In: 34th International
  Conference on Computer Aided Verification. Lecture Notes in Computer Science,
  vol. 13372, pp. 174--187. Springer (2022). \doi{10.1007/978-3-031-13188-2\_9}

\bibitem{DziembowskiJW97}
Dziembowski, S., Jurdzinski, M., Walukiewicz, I.: How much memory is needed to
  win infinite games? In: 12th Annual {IEEE} Symposium on Logic in Computer
  Science. pp. 99--110. {IEEE} Computer Society (1997).
  \doi{10.1109/LICS.1997.614939}

\bibitem{Ehlers11a}
Ehlers, R.: Generalized rabin(1) synthesis with applications to robust system
  synthesis. In: Third International Symposium on {NASA} Formal Methods.
  Lecture Notes in Computer Science, vol.~6617, pp. 101--115. Springer (2011).
  \doi{10.1007/978-3-642-20398-5\_9}

\bibitem{Ehlers11}
Ehlers, R.: Unbeast: Symbolic bounded synthesis. In: 17th International
  Conference on Tools and Algorithms for the Construction and Analysis of
  Systems. Lecture Notes in Computer Science, vol.~6605, pp. 272--275. Springer
  (2011). \doi{10.1007/978-3-642-19835-9\_25}

\bibitem{GiacomoV15}
Giacomo, G.D., Vardi, M.Y.: Synthesis for {LTL} and {LDL} on finite traces. In:
  Yang, Q., Wooldridge, M.J. (eds.) Twenty-Fourth International Joint
  Conference on Artificial Intelligence. pp. 1558--1564. {AAAI} Press (2015)

\bibitem{HausmannSchroeder21}
Hausmann, D., Schr{\"{o}}der, L.: Quasipolynomial computation of nested
  fixpoints. In: Groote, J.F., Larsen, K.G. (eds.) Tools and Algorithms for the
  Construction and Analysis of Systems - 27th International Conference, {TACAS}
  2021, Held as Part of the European Joint Conferences on Theory and Practice
  of Software, {ETAPS} 2021, Luxembourg City, Luxembourg, March 27 - April 1,
  2021, Proceedings, Part {I}. Lecture Notes in Computer Science, vol. 12651,
  pp. 38--56. Springer (2021). \doi{10.1007/978-3-030-72016-2\_3},
  \url{https://doi.org/10.1007/978-3-030-72016-2\_3}

\bibitem{HunterD05}
Hunter, P., Dawar, A.: Complexity bounds for regular games. In: 30th
  International Symposium on Mathematical Foundations of Computer Science.
  Lecture Notes in Computer Science, vol.~3618, pp. 495--506. Springer (2005).
  \doi{10.1007/11549345\_43}

\bibitem{JohnJBK21}
John, T., Jantsch, S., Baier, C., Kl{\"{u}}ppelholz, S.: Determinization and
  limit-determinization of {Emerson-Lei} automata. In: 19th International
  Symposium on Automated Technology for Verification and Analysis. Lecture
  Notes in Computer Science, vol. 12971, pp. 15--31. Springer (2021).
  \doi{10.1007/978-3-030-88885-5\_2}

\bibitem{JohnJBK22}
John, T., Jantsch, S., Baier, C., Kl{\"{u}}ppelholz, S.: From emerson-lei
  automata to deterministic, limit-deterministic or good-for-mdp automata.
  Innov. Syst. Softw. Eng.  \textbf{18}(3),  385--403 (2022).
  \doi{10.1007/s11334-022-00445-7},
  \url{https://doi.org/10.1007/s11334-022-00445-7}

\bibitem{Kress-GazitFP09}
Kress{-}Gazit, H., Fainekos, G.E., Pappas, G.J.: Temporal-logic-based reactive
  mission and motion planning. {IEEE} Trans. Robotics  \textbf{25}(6),
  1370--1381 (2009). \doi{10.1109/TRO.2009.2030225}

\bibitem{KuglerS09}
Kugler, H., Segall, I.: Compositional synthesis of reactive systems from live
  sequence chart specifications. In: 15th International Conference on Tools and
  Algorithms for the Construction and Analysis of Systems. Lecture Notes in
  Computer Science, vol.~5505, pp. 77--91. Springer (2009).
  \doi{10.1007/978-3-642-00768-2\_9},
  \url{https://doi.org/10.1007/978-3-642-00768-2\_9}

\bibitem{kupferman2001model}
Kupferman, O., Vardi, M.Y.: Model checking of safety properties. Formal methods
  in system design  \textbf{19}(3),  291--314 (2001).
  \doi{10.1023/A:1011254632723}

\bibitem{LichtensteinPZ85}
Lichtenstein, O., Pnueli, A., Zuck, L.D.: The glory of the past. In: Conference
  on Logics of Programs. Lecture Notes in Computer Science, vol.~193, pp.
  196--218. Springer (1985). \doi{10.1007/3-540-15648-8\_16},
  \url{https://doi.org/10.1007/3-540-15648-8\_16}

\bibitem{LiuOTM13}
Liu, J., Ozay, N., Topcu, U., Murray, R.M.: Synthesis of reactive switching
  protocols from temporal logic specifications. {IEEE} Trans. Autom. Control.
  \textbf{58}(7),  1771--1785 (2013). \doi{10.1109/TAC.2013.2246095}

\bibitem{MajumdarS23}
Majumdar, R., Schmuck, A.: Supervisory controller synthesis for nonterminating
  processes is an obliging game. {IEEE} Trans. Autom. Control.  \textbf{68}(1),
   385--392 (2023). \doi{10.1109/TAC.2022.3143108}

\bibitem{MoarrefK20}
Moarref, S., Kress{-}Gazit, H.: Automated synthesis of decentralized
  controllers for robot swarms from high-level temporal logic specifications.
  Auton. Robots  \textbf{44}(3-4),  585--600 (2020).
  \doi{10.1007/s10514-019-09861-4}

\bibitem{MuellerSickert17}
M{\"{u}}ller, D., Sickert, S.: {LTL} to deterministic emerson-lei automata. In:
  Bouyer, P., Orlandini, A., Pietro, P.S. (eds.) Proceedings Eighth
  International Symposium on Games, Automata, Logics and Formal Verification,
  GandALF 2017, Roma, Italy, 20-22 September 2017. {EPTCS}, vol.~256, pp.
  180--194 (2017). \doi{10.4204/EPTCS.256.13},
  \url{https://doi.org/10.4204/EPTCS.256.13}

\bibitem{PitermanP06}
Piterman, N., Pnueli, A.: Faster solutions of rabin and streett games. In: 21th
  {IEEE} Symposium on Logic in Computer Science {(LICS} 2006), 12-15 August
  2006, Seattle, WA, USA, Proceedings. pp. 275--284. {IEEE} Computer Society
  (2006). \doi{10.1109/LICS.2006.23}

\bibitem{PitermanPS06}
Piterman, N., Pnueli, A., Sa'ar, Y.: Synthesis of reactive(1) designs. In: 7th
  International Conference on Verification, Model Checking, and Abstract
  Interpretation. Lecture Notes in Computer Science, vol.~3855, pp. 364--380.
  Springer (2006). \doi{10.1007/11609773\_24}

\bibitem{PnueliR89}
Pnueli, A., Rosner, R.: On the synthesis of a reactive module. In: Sixteenth
  {ACM} Symposium on Principles of Programming Languages. pp. 179--190. {ACM}
  Press (1989). \doi{10.1145/75277.75293}

\bibitem{RamanPFK15}
Raman, V., Piterman, N., Finucane, C., Kress{-}Gazit, H.: Timing semantics for
  abstraction and execution of synthesized high-level robot control. {IEEE}
  Trans. Robotics  \textbf{31}(3),  591--604 (2015).
  \doi{10.1109/TRO.2015.2414134},
  \url{https://doi.org/10.1109/TRO.2015.2414134}

\bibitem{RenkinDP20}
Renkin, F., Duret{-}Lutz, A., Pommellet, A.: Practical "paritizing" of
  emerson-lei automata. In: Hung, D.V., Sokolsky, O. (eds.) Automated
  Technology for Verification and Analysis - 18th International Symposium,
  {ATVA} 2020, Hanoi, Vietnam, October 19-23, 2020, Proceedings. Lecture Notes
  in Computer Science, vol. 12302, pp. 127--143. Springer (2020).
  \doi{10.1007/978-3-030-59152-6\_7},
  \url{https://doi.org/10.1007/978-3-030-59152-6\_7}

\bibitem{SohailS13}
Sohail, S., Somenzi, F.: Safety first: a two-stage algorithm for the synthesis
  of reactive systems. Int. J. Softw. Tools Technol. Transf.  \textbf{15}(5-6),
   433--454 (2013). \doi{10.1007/s10009-012-0224-3}

\bibitem{TabakovV05}
Tabakov, D., Vardi, M.Y.: Experimental evaluation of classical automata
  constructions. In: 12th International Conference on Logic for Programming,
  Artificial Intelligence, and Reasoning. Lecture Notes in Computer Science,
  vol.~3835, pp. 396--411. Springer (2005). \doi{10.1007/11591191\_28},
  \url{https://doi.org/10.1007/11591191\_28}

\bibitem{Thomas02}
Thomas, W.: Infinite games and verification (extended abstract of a tutorial).
  In: 14th International Conference on Computer Aided Verification. Lecture
  Notes in Computer Science, vol.~2404, pp. 58--64. Springer (2002).
  \doi{10.1007/3-540-45657-0\_5}

\bibitem{VardiS85}
Vardi, M.Y., Stockmeyer, L.J.: Improved upper and lower bounds for modal logics
  of programs: Preliminary report. In: Proceedings of the 17th Annual {ACM}
  Symposium on Theory of Computing. pp. 240--251. {ACM} (1985)

\bibitem{WongpiromsarnTM12}
Wongpiromsarn, T., Topcu, U., Murray, R.M.: Receding horizon temporal logic
  planning. {IEEE} Trans. Autom. Control.  \textbf{57}(11),  2817--2830 (2012).
  \doi{10.1109/TAC.2012.2195811}

\bibitem{ZhuTLPV17}
Zhu, S., Tabajara, L.M., Li, J., Pu, G., Vardi, M.Y.: A symbolic approach to
  safety {LTL} synthesis. In: 13th International Haifa Verification Conference:
  Hardware and Software - Verification and Testing. Lecture Notes in Computer
  Science, vol. 10629, pp. 147--162. Springer (2017).
  \doi{10.1007/978-3-319-70389-3\_10},
  \url{https://doi.org/10.1007/978-3-319-70389-3\_10}

\bibitem{abs-2008-06790}
Zhu, S., Tabajara, L.M., Pu, G., Vardi, M.Y.: On the power of automata
  minimization in temporal synthesis. In: Proceedings 12th International
  Symposium on Games, Automata, Logics, and Formal Verification. {EPTCS},
  vol.~346, pp. 117--134 (2021). \doi{10.4204/EPTCS.346.8},
  \url{https://doi.org/10.4204/EPTCS.346.8}

\bibitem{Zielonka98}
Zielonka, W.: Infinite games on finitely coloured graphs with applications to
  automata on infinite trees. Theor. Comput. Sci.  \textbf{200}(1-2),  135--183
  (1998). \doi{10.1016/S0304-3975(98)00009-7},
  \url{https://doi.org/10.1016/S0304-3975(98)00009-7}

\end{thebibliography}

\end{document}